\documentclass[12pt]{article}

\usepackage{a4wide}
\usepackage{amsmath,amsfonts,
amssymb, amsthm,dsfont,color}
\usepackage{hyperref,a4wide}
\usepackage{csquotes}

\usepackage{mathabx}%fancy inequalities
\usepackage{dirtytalk}
\usepackage[active]{srcltx}

\usepackage{mathrsfs}
\usepackage{amscd,graphicx}

\usepackage{tikz}
\usepackage{pgfplots}
\pgfplotsset{compat=1.18}

\usepackage[backend=bibtex,
  bibencoding=utf8,
  style=bibstyle
  ]{biblatex}
\newtheorem{theorem}{Theorem}[section]
\newtheorem{definition}[theorem]{Definition}
\newtheorem{proposition}[theorem]{Proposition}
\newtheorem{corollary}[theorem]{Corollary}
\newtheorem{lemma}[theorem]{Lemma}

\theoremstyle{definition}
\newtheorem{remark}[theorem]{Remark}
\newtheorem{condition}[theorem]{Condition}

\def\Z{\mathbb{Z}}

\def\R{\mathbb{R}}

\def\vep{\varepsilon} 
\def\C{{\mathbb C}} 
\def\N{\mathbb{N}}

\def\O{\mathbf{O}}

\def\bb1{{\rm{1}\hspace{-3pt}\mathbf{l}}}
\def\Ie0{[-\varepsilon_0,\varepsilon_0]}

\def\Id{\mathbf{1}}

\def\Tr{{\rm Tr}}

\def\BC2{\mathbb{B}\big(\mathbb{C}^2\big)}

\def\beq{\begin{equation}}
\def\eeq{\end{equation}}

\numberwithin{equation}{section}

\title{%{\color{red}MHT suggestion: Almost invariant subspaces and magnetic perturbations for effect Hofstadter-like Hamiltonians 

Almost invariant subspaces in the topologically non-trivial setting: a pedestrian approach}

\author{H.~D. Cornean\footnote{Department of Mathematical Sciences, Aalborg University, Thomas Manns Vej 23, 9220 Aalborg, Denmark; e-mail: cornean@math.aau.dk}, G.~Peluso\footnote{Department of Mathematical Sciences, Aalborg University, Thomas Manns Vej 23, 9220 Aalborg, Denmark; e-mail: gpe@math.aau.dk}, and M.~H. Thorn\footnote{Department of Mathematical Sciences, Aalborg University, Thomas Manns Vej 23, 9220 Aalborg, Denmark; e-mail: mikkelht@math.aau.dk}}

\begin{document}

\maketitle

\begin{abstract}
    We present a self-contained, algorithmic construction of almost invariant subspaces \cite{PanatiSpohnTeufel2003, DeNittisPanati2010, FreundTeufel2016} for Hofstadter-like Hamiltonians $H_b $ acting on $\ell^2(\Z^2,\C^n)$. More precisely, for rational values of $b_0/(2\pi)$ we describe an iterative procedure that takes a ``locally gapped" projection $P_{0}$ (possibly topologically non-trivial) which commutes with $H_{b_0}$, into a projection $P_{\varepsilon}$ that commutes with $H_{b_0+\varepsilon}$ up to an error of order $\varepsilon^\infty$. 
    
    Our construction insures that the Chern number of $P_{\varepsilon}$ equals the Chern number of $P_0$ as long as $P_\vep$ exists. Moreover, its integrated density of states varies linearly with $\vep$ (the ``St\v reda formula"). We also show that any spatially localized Parseval frame of $P_{0}$ can be mapped into to a localized Parseval frame of $P_{\varepsilon}$. 

    We choose to work with a discrete model in order to keep the technicalities to a minimum level, with no need for advanced pseudo-differential calculus. Nevertheless, we are confident  that most of our results can be extended to (unbounded) magnetic (pseudo)differential operators, but the technical burden would be much heavier. This is left for future work.
\end{abstract}

\tableofcontents

\section{Introduction}

Two-dimensional periodic quantum systems subject to a perpendicular magnetic field provide a setting in which spectral theory, transport, and topology are deeply intertwined. At the physical level, the integer quantum Hall effect reveals the remarkable stability of the transverse conductivity under perturbations of the underlying Hamiltonian. At the mathematical level, this stability is encoded by an integer-valued topological invariant associated with the occupied states, namely the first Chern number/character of the Fermi projection. The same invariant enters several apparently different descriptions of the system: it determines the Hall conductivity, controls the topological obstruction to the existence of localized Wannier bases, and equals up to a universal constant the variation slope of the integrated density of states with respect to the magnetic field through the Středa formula \cite{TKNN1982, Streda1982, BellissardVanElstSchulzBaldes1994, CorneanMonacoMoscolari2021}.

The spectral analysis of magnetic periodic operators presents difficulties that cannot be treated with analytic perturbation theory. A perturbation coming from a constant magnetic field is typically not relatively bounded to the unperturbed operator, not even in the discrete setting. Moreover, the Hamiltonian admits a Bloch-Floquet decomposition only for rational values of the magnetic flux through a fundamental cell, making the irrational fluxes harder to study. This issue is especially relevant for  Hofstadter operators, where the spectrum may exhibit an intricate dependence on the magnetic flux \cite{Hofstadter1976}.

The first major breakthrough in treating small enough magnetic perturbations in the continuum case \cite{Nenciu1991} dealt with the case of an unperturbed simple Bloch band which is separated from the rest of the spectrum by ``genuine" gaps, which means that its range does not overlap with the range of other bands. Here, the associated Riesz projection $P_0$ is a proper spectral projector of the unperturbed operator $H_0$. It turns out that the spectral gaps \cite{Nenciu2002, CorneanHerbstNenciu2016} are stable when the external magnetic perturbation is small enough and one can still define a spectral Riesz projection $P_\varepsilon$ which commutes with the perturbed magnetic Schr\"odinger operator $H_\varepsilon$.  Under the assumption that the range of $P_0$ admits a basis of strongly localized Wannier functions, then one can also construct a strongly localized basis for $P_\varepsilon$ indexed by the periodicity lattice $\Gamma$ of $P_0$ and one can define a Hofstadter-like effective matrix acting on $\ell^2(\Gamma)$ which is unitarily equivalent with the original magnetic band operator $H_\varepsilon P_\varepsilon$ \cite{Nenciu1991, CorneanHerbstNenciu2016}. 

The same procedure can be extended to a group of possibly overlapping/crossing Bloch eigenvalues whose union of ranges still remains isolated from the others, by first constructing composite Wannier functions \cite{Panati2007, FiorenzaMonacoPanati2016, MonacoPanatiPisanteTeufel2018} when they exist, then add the small magnetic perturbation as in \cite{CorneanHerbstNenciu2016, CorneanMonacoMoscolari2019}. 

A technically more challenging problem is the one in which the group of eigenvalues we are interested in remains isolated from the rest at each quasi-momentum, but their union of ranges overlaps with the ranges of the other eigenvalues. This is the so-called ``locally gapped case". Here, $P_{0}$ is still an orthogonal projection which commutes with $H_0$,  but it is no longer a spectral projection for $H_{0}$. In this case the relevant question is whether one can construct an \emph{almost-invariant subspace} that retains the properties of an invariant subspace up to correction terms that decay with the intensity of the magnetic field \cite{PanatiSpohnTeufel2003, DeNittisPanati2010, FreundTeufel2016}.

In our paper, we provide an algorithmic and rather elementary procedure to construct such invariant subspaces for a certain family of discrete Hamiltonians $H_{b_0+\varepsilon} $ acting on $\ell^2(\Z^2,\C^n)$ when $b_0/(2\pi)$ is a rational number.
We describe an iterative procedure that takes a ``locally gapped" projection $P_{0}$ which commutes with $H_{b_0}$, into a projection $P_{\varepsilon}$ that commutes with $H_{b_0+\varepsilon}$ up to an error of order $\varepsilon^\infty$. Our construction insures that the Chern number of $P_{\varepsilon}$ is constant and equal (up to a universal constant)  to the derivative with respect to $\varepsilon$ of its integrated density of states.

We remind the reader that if the Chern number of $P_0$ is different from zero, then $P_0$ does not admit an exponentially localized Wannier basis, while a localized Parseval frame always exists \cite{CorneanMonacoMoscolari2019,AucklyKuchment2018}. An important consequence of our method is that any exponentially  localized Parseval frame of $P_{0}$ can be mapped by a canonical procedure into an exponentially localized Parseval frame of $P_{\varepsilon}$, without asking for a ``strict spectral gap" like in \cite{CorneanMonacoMoscolari2019}.

%We remark that the discrete setting does not originate from a discretization of a continuum model as $n$ represents the number of vectors that generate the Wannier basis with their translations. So $H_\varepsilon $ are unitarily equivalent to the restrictions of the perturbed Hamiltonians to the selected Bloch band. 

%\textcolor{red}{This part should be revised by CHIARISSIMO Professor Cornean} In full generality the starting effective Hamiltonian would be $H_{0,\varepsilon}$ with a smooth dependence on $\varepsilon$. However, any properties on this more general case are easily deducible form the contents of this paper.

\subsection{Setting and the main results}

Let us consider the space of functions $f\colon\Z^2 \to \C^n $ and denote the elements of $\Z^2$ by the letter $\gamma$ and the elements of $\C^n$ by the letter $v$. 
Consider the scalar product $$\left\langle g,f \right\rangle = \sum_{\gamma \in \Z^2}  \left\langle g(\gamma )  | f(\gamma)\right\rangle_{\C^n}.$$ Then, the functions with finite norm compose the Hilbert space $\mathcal{H}=\ell^2( \Z^{2}, \C^n)$.
If $\{e_a\}_{a\in \{1,\dots ,n\}}$ is the canonical basis of $\C^n$, then the following is a basis of $\mathcal{H}$:
\begin{equation*}\label{eq:basis_of_H}\{\delta_{\gamma,a}\}_{\gamma\in\Z^2 ,a\in \{1,\dots ,n\}} \subset \mathcal{H} \quad \mid  \quad \delta_{\gamma,a}  ( \gamma')=\begin{cases}
    e_a & \mbox{if }\gamma = \gamma', \\ 0 & \mbox{otherwise.}
\end{cases}\end{equation*}
Any operator $A\colon\mathcal{H}\to \mathcal{H}$ is uniquely determined by the set of \emph{hopping matrices} $A(\gamma,\gamma') \in M(n\times n, \C)$ with \emph{hopping amplitudes} $$[A(\gamma,\gamma')]_{a,b} = \left\langle \delta_{\gamma,a}|A(\delta_{\gamma ',b} )\right\rangle, \quad \forall\gamma,\gamma'\in\Z^2,\quad a,b\in\{1,\dots , n\}. $$
They satisfy $$A(f)(\gamma)=\sum_{\gamma'\in\Z^2}A(\gamma,\gamma')f(\gamma'), \quad \forall \gamma\in\Z^2.$$
We denote by $\mathcal{B}(\mathcal{H})$ the set of bounded operators  with the operator norm topology. We are also 
interested in operators with hopping matrices that are strongly localized near the diagonal:

\begin{definition}[Locality] \label{def:localization}
    
    We say that an operator $A\colon\mathcal{H}\to \mathcal{H}$ is exponentially localized near the diagonal if there exist two constants $w,C>0$ such that its hoppings obey the constraint $$ \| A(\gamma,\gamma') \| \le C e^{-w |\gamma-\gamma'|}. $$
    Throughout this paper we refer to these operators simply as \emph{local}, we call $w$ the \emph{exponential decay rate}, and denote the set of local operators as $\mathcal{B}^{\exp}(\mathcal{H})$.
\end{definition}

\begin{definition}[Exponential asymptotic approximations]
   A family of local operators $A_\varepsilon \in \mathcal{B}^{\exp}(\mathcal{H})$ will be denoted by  $A_\varepsilon=\mathcal{O}^{\exp}(\varepsilon^k)$ if there exist two constants $w,C>0$, both  independent of $\vep$, such that, for $\vep$ in a neighborhood of zero, the hopping matrices obey the estimate $$\|A_\varepsilon(\gamma,\gamma')\|<C|\varepsilon|^k e^{-w|\gamma-\gamma'|}.$$ 
   
\end{definition}

In Appendix \ref{sec:exp_localization} we prove that sums, compositions, resolvents, and functional calculus are all closed operations in $\mathcal{B}^{\exp}(\mathcal{H})$.
However, there is no closure with respect to convergence using the operator norm, despite having $\mathcal{B}^{\exp}(\mathcal{H})\subset \mathcal{B}(\mathcal{H})$. 
%We also prove that the symbol $\mathcal{O}^{\exp}$ satisfies the usual properties such as $\mathcal{O}^{\exp}(\varepsilon^k) \mathcal{O}^{\exp}(\varepsilon^l)=\mathcal{O}^{\exp}(\varepsilon^{k+l})$.

Now we describe how the magnetic field is introduced.
\begin{definition}[Peierls phase]
    {We call Peierls phase the following} bilinear and skew-symmetric function $\varphi\colon\Z^2\times\Z^2 \to \R$ such that $\varphi(\gamma,\gamma')=\frac{1}{2}\big( \gamma_2\gamma_1'- \gamma_1\gamma_2'\big).$
\end{definition}

\begin{definition}[Magnetic translations]
    Given $\varepsilon\in\R$ and $\gamma''\in\Z^2$, define the unitary operator $\tau_{\varepsilon,\gamma''}\colon\mathcal{H}\rightarrow\mathcal{H}$ by $(\tau_{\vep,\gamma''}f)(\gamma)=e^{i\vep\varphi(\gamma,\gamma'')}f(\gamma-\gamma'').$
    %\quad \Leftrightarrow\quad \tau_{\vep,\gamma''}(\gamma,\gamma')= \begin{cases}
       % e^{i\vep\varphi(\gamma,\gamma'')}\Id &\mbox{if }\gamma'=\gamma-\gamma'',\\0&\mbox{otherwise.}  \end{cases}
       
    Given $q\in\N\setminus\{0\}$, an operator $A$ is called $(\varepsilon,q)$-periodic if $$\left[A,\tau_{\varepsilon,(q\alpha_1,\alpha_2)}\right]=0, \quad \forall(\alpha_1,\alpha_2)=\alpha\in\Z^2.$$
\end{definition}

\begin{definition}[The ``discrete magnetic quantization" $\O_\varepsilon$]\label{def:Magnetization}
    For any local operator $A\in \mathcal{B}^{\exp}(\mathcal{H})$ and any $\varepsilon\in\mathbb{R}$, we define $\O_\varepsilon (A)$ through its hopping matrices: $$\O_\varepsilon(A)(\gamma,\gamma') = e^{i\varepsilon \varphi (\gamma ,\gamma')} A(\gamma,\gamma').$$
\end{definition}

The previously introduced objects have several properties which will be proved at the beginning of Section \ref{sec:magnetic_calculus}.

\begin{lemma}\label{lem:properties_O_epsilon}
    For any $\varepsilon,\eta\in\R, q\in \N\setminus \{0\}$ and $A\in     \mathcal{B}^{\exp}(\mathcal{H}) $, we have that:
    \begin{itemize}
        \item %$\O_\varepsilon$ preserves locality with identical exponential weights, so 
        $A\mapsto \O_\varepsilon(A)$ is linear, $\O_\varepsilon(    \mathcal{B}^{\exp}(\mathcal{H})) \subset \mathcal{B}^{\exp}(\mathcal{H})$
        %\item $\O_\varepsilon$ is \emph{linear}: $\O_\varepsilon(\lambda A)=\lambda\O_\varepsilon(A),\quad \O_\varepsilon(A_1+A_2)=\O_\varepsilon(A_1)+\O_\varepsilon(A_2). $  
        and %$\O$ respects the group structure of $\R$: 
        $\O_\eta(\O_\varepsilon(A)) =\O_{\eta+\varepsilon}(A).$
        \item $A$ is self-adjoint if and only if $\O_\varepsilon(A)$ is self-adjoint.        
        \item $A$ is $(\varepsilon,q)$-periodic if and only if $\O_\eta(A)$ is $(\varepsilon+\eta,q)$-periodic.
        \end{itemize}
\end{lemma}

%Using this language, a local Hamiltonian that commutes with all canonical translations is $(0,1)$-periodic, and its perturbed Hofstadter-like Hamiltonian  $H_\varepsilon=\O_\varepsilon(H_0)$ is local and $(\varepsilon,1)$-periodic.
In this language, if $b_0=\frac{2\pi p}{q}$ then $H_{b_0+\varepsilon}=\O_\varepsilon(H_{b_0})$, which transforms a $(b_0,1)$-periodic Hamiltonian into a $(b_0+\varepsilon,1)$-periodic perturbed one.
Moreover, being $(\varepsilon,1)$-periodic implies being $(\varepsilon,q)$-periodic for all $q\in \N\setminus\{0\}$.

Now we can formulate the main result of our paper. 
\begin{theorem}\label{thm:main_theorem}
    Consider a $(0,1)$-periodic self-adjoint operator $H_0\in \mathcal{B}^{\exp}(\mathcal{H})$ and set $H_b=\O_b(H_0)$ with $b\in \R$.
    Fix $b_0=\frac{2\pi p}{q}$ with any two co-prime integers $p,q$ with $q>0$ and suppose that $H_{b_0}$ either has a local spectral island or a strict spectral island (the definitions of these are given in Appendix \ref{sec:fiber_decomposition}). Then:
    
    \noindent{\rm (i).} There exists a local, orthogonal  projection $P_0 \in \mathcal{B}^{\exp}(\mathcal{H})$ associated with the (local) spectral island. It is $(b_0,q)$-periodic and commutes with $H_{b_0}$. It is not a spectral projection of $H_{b_0}$ unless we assume the strict gap condition.
    
    \noindent{\rm (ii).} There exists $\vep_1>0$ such that for all $0<|\vep|<\vep_1$ one can construct an orthogonal  projection $P_\vep=P_\vep^2=P_\vep^*\in  \mathcal{B}^{\exp}(\mathcal{H})$ which is $(b_0+\varepsilon,q)$-periodic and with 
   $$\lim_{\vep\rightarrow 0}\|P_\vep-\O_\vep(P_0)\|= 0 ,\quad \mbox{and} \quad \lim_{\vep\rightarrow 0}\Big (|\vep|^{-l}\Vert [H_{b_0+\vep},P_\vep]\Vert \Big ) =0, \quad \forall l \in \N .$$
   %To be more specific, $\varepsilon_1$ is a value for which $\|\O_\varepsilon^2(P_0)-\O_\varepsilon(P_0)\|<1/8$ for all $|\varepsilon|< \varepsilon_1$.
The exponential decay rate of $P_\vep$ may become weaker and weaker when $\vep\to 0$ in the case of a local spectral island.
   
   \noindent{\rm (iii).} Under the strict gap condition, $P_\vep$ is a spectral Riesz projection of $H_{b_0+\vep}$ (so they commute).  
   
   \noindent{\rm (iv).} Let $X_j$ be the multiplication operator with $\gamma_j$. The Chern number of $P_\vep$ given by 
   $${\rm Ch}(P_\vep)=2\pi i\lim_{L\to\infty} \frac{1}{(2L)^2} \sum_{\gamma\in [-L,L]^2} \Tr \Big (P_\vep \big [\, [X_1,P_\vep],[X_2,P_\vep]\big ]\Big )(\gamma,\gamma) $$
   is an integer independent of $\vep$ and equals ${\rm Ch}(P_0)$.

   \noindent{\rm (v).} The ``St\v{r}eda formula" \cite{CorneanMonacoMoscolari2021} holds, which means that the integrated density of states
   $$\mathcal{I}(P_\varepsilon)=\lim_{L\to \infty} \frac{1}{(2L)^2} \sum_{\gamma\in [-L,L]^2} \Tr (P_\varepsilon(\gamma,\gamma))$$
   varies linearly in $\vep$ and %up to a finite number of points, with a slope given by: 
   %$$2\pi\frac{d}{d\vep} \mathcal{I}(P_\varepsilon)= {\rm Ch}(P_\vep)\equiv{\rm Ch}(P_0).$$
   $$\mathcal{I}(P_\varepsilon)=\mathcal{I}(P_0)+\frac{{\rm Ch}(P_0) }{2\pi}\, \varepsilon.$$

   %\noindent{\rm (vi).} If ${\rm Ch}(P_0)\neq 0$, then $\lim_{\eta\to\vep}\Vert P_{\eta}-P_{\varepsilon}\Vert =1$. Thus $P_\vep$ is nowhere continuous in $\vep$ with respect to the operator norm topology. 

   \noindent{\rm (vi).} If there exists a set $\{\Psi_j\}_{j\in\{1,\dots ,m \}} \subset \mathcal{H}$ of exponentially localized vectors such that the collection $\{\tau_{b_0,\gamma}\Psi_j\}^{\gamma \in q\Z\times \Z}_{j\in\{1,\dots ,m \}}$ forms a Parseval frame of $P_0$, see \eqref{oct15}, then for $\varepsilon$ near zero there exists a family of self-adjoint operators $B_\varepsilon \in \mathcal{B}^{\exp}(\mathcal{H})$ and one of unitary operators $U_\varepsilon \in \mathcal{B}^{\exp}(\mathcal{H})$ such that the collection $\tau_{b_0+\varepsilon,\gamma}U_\varepsilon B_\varepsilon\Psi_j$ forms an exponentially localized Parseval frame of $P_\varepsilon$ for $j\in\{1,\dots ,m \}$ and $\gamma\in q\Z\times \Z$.
\end{theorem}

\subsection{The structure of the paper}

For clarity purposes, we prove each point of the main theorem in a dedicated section.
Point \rm{(i)} is not original and is solved using standard techniques, so we dedicate Appendix \ref{sec:fiber_decomposition} to it (Lemma \ref{lem:P_0_local}). There, we also explain in detail the difference between a local spectral island and a strict spectral island. 
Since this paper focuses on local operators, we summarize their properties in Appendix \ref{sec:exp_localization}.

In Section \ref{sec:magnetic_calculus} we analyze the properties of the ``discrete magnetic quantization" $\O_\varepsilon$ that will be repeatedly used throughout the rest of the paper.

Section \ref{sec:P_epsilon_construction} is the most prominent one as it contains the proofs of points \rm{(ii)} (Proposition \ref{pro:construction_P_epsilon}) and \rm{(iii)} (Lemma \ref{lem:P_epsilon_strictly_gapped}). The construction procedure permits us to explicitly build up  $P_\varepsilon$ in a rather down-to-earth manner. 

Section \ref{sec:Chern_numbers} deals with the proof of point \rm{(iv)} of the theorem (Proposition \ref{pro:Ch_constant}). The challenging issue here is to prove continuity in $\vep$, while the constancy of the Chern number immediately follows by the fact that it is integer-valued. 

Section \ref{sec:Streda_formula} contains the proof of the St\v{r}eda formula of point \rm{(v)} (Proposition \ref{pro:derivative_of_IDS_of_P_epsilon}).
%and the proof of point \rm{(vi)} (Corollary \ref{cor:P_epsilon_discontinuous}).

Section \ref{sec:Parseval_frames} contains the proof of point \rm{(vi)}, which is the construction of an exponentially localized Parseval frame for $P_\vep$.  
%Once again, the proof contains an explicit procedure to construct $B_\varepsilon $ and $Q_\varepsilon$, implicitly stating that a Parseval frame for $P_\varepsilon$ can be computed algorithmically. 

%{Lastly, in Appendix \ref{sec:rescaling} we sketch another way of approaching the problem of perturbations near rational fluxes: instead of using $(\varepsilon,q)$-periodicity one can rescale the unit cell of the lattice using unitary operators.}

\subsection{Comparison with the existing literature}

Several authors have had major contributions during the the past few decades. We have already mentioned \cite{Nenciu1991, CorneanHerbstNenciu2016} concerning the strictly gapped case, which also hold in the continuum case.

Also in the continuum regime and with a local spectral island,  \cite{PanatiSpohnTeufel2003} is probably the first work which gives a description of the effective dynamics and constructs almost invariant subspaces near $\varepsilon =0$ for magnetic fields coming from slowly varying smooth and bounded magnetic vector potentials. The paper \cite{FreundTeufel2016} allows constant magnetic field variations near any $b_0=\frac{2\pi p}{q}$, and it is based on a rather sophisticated pseudo-differential calculus, especially when the existence of an initial strongly localized Wannier basis is not guaranteed ($\mathrm{Ch}(P_0)\ne 0$). 

A distinctive feature of our construction is its explicit character. The projection $P_\varepsilon$ is obtained from concrete operations on the reference projection and its localized kernel, rather than through an abstract existence theorem.
The operations applied to the hopping matrices are: the application of a magnetic phase generated by the change of magnetic flux, the corrections required to restore idempotency of the projection, and the elimination of the leading term of $[P_\varepsilon,H_\varepsilon]$.

Our formulation is suitable for finite-volume truncations and numerical implementations. For example, when $b_0=0$ and $\vep =2\pi/m$ with a larger and larger $m$, a  natural but not so fortunate idea would be to rescale the unit cell with $m$ at each step,  and  diagonalize an increasingly large magnetic Bloch matrix.  Our method avoids this. 

%{\color{red} While the conservation of the Chern number is nowhere original, the validity of the Středa formula in this regime may be interesting as it obtains the same result present in the continuum setting \cite{CorneanMonacoMoscolari2021}.}

At last, we believe that our most important ``applied" contribution is the possibility to turn any localized Parseval frame of $P_0$ into one of $P_\varepsilon$, which could enhance the possibility to numerically analyze effective Hamiltonians with arbitrary magnetic fluxes, as explicit methods for constructing Parseval frames for topologically non-trivial projections have already been developed \cite{CorneanMonacoMoscolari2019, FerreriMonacoPeluso2026}.

\subsection{Open questions and outlook}
We list here a few issues which remain open:
\begin{enumerate}
    \item Would it be possible to repeat the construction of $P_\vep$ in  continuum models, where $H_b$ is replaced by a family of unbounded, elliptic magnetic (pseudo)differential operators as in \cite{Purice2004,Purice2007,Purice2019,CorneanGarde2019,Thorn2026}? Most probably yes, but the technical burden would be much heavier.
    \item There has been a lot of recent activity related to ``non-equilibrium almost stationary states" (NEASS) \cite{Teufel2020} and their applications to quantum Hall effect. For example, in \cite{DomenicoGiovanna2022,MazziniMonaco2026} it is shown that under the strict gap condition, the Hall current is given by the linear response in the constant external electric field $E$, up to a correction of order $E^\infty$. The coefficient of the linear response is nothing but the Chern number of the unperturbed Fermi projection. This suggests that there might exist a connection between our integrated density of states $\mathcal{I}(P_\vep)$, and the Hall current.  
    \item Given any real number $b$, Dirichlet's approximation theorem states that there exists an integer $p/q$ such that $|b-p/q|\leq 1/q^2$. Our construction of $P_\vep$ near any given rational $b_0$ provides us with a maximal $\vep_1>0$ such that $P_\vep$ exists for $|\vep|<\vep_1$. The challenge is that $\vep_1$ has a complicated dependence on the denominator $q$ and on the minimal value of the local spectral gap, and it would be interesting to see if $q^2\vep_1$ has a lower bound or not.  In any case, having a more explicit $q$-dependence of $\vep_1$ would tell us which type of irrational fluxes we can reach by starting from a nearby rational one. 
\end{enumerate}

\subsection{Acknowledgements}
The authors acknowledge support from the Independent Research Fund
Denmark--Natural Sciences, grant DFF–10.46540/2032-00005B, and from DNRF Centre
CLASSIQUE sponsored by the Danish National Research Foundation, grant nr. 187.

The authors would also like to thank Domenico Monaco and Gianluca Panati for valuable discussions.

\section{The discrete magnetic calculus}\label{sec:magnetic_calculus}

The properties of $\O_\varepsilon$ are pivotal for the proof of the main results, and in this section we investigate them.

\begin{proof}[Proof of Lemma \ref{lem:properties_O_epsilon}]

%The first statement follows from the equality $$\|\O_\varepsilon(A)(\gamma,\gamma')\|=\|A(\gamma,\gamma')\|.$$
%
%The second and third statements are trivial to prove.
%
%For the fourth property, observe that 
%\[\begin{split}A=A^* & \Leftrightarrow [A(\gamma,\gamma')]_{a,b} = [A^*(\gamma,\gamma')]_{a,b}, \quad \forall \gamma,\gamma'\in\Z^2,\quad a,b\in\{1,\dots,n\}.\end{split}\]
%Since $$[A^*(\gamma,\gamma')]_{a,b} = \left\langle \delta_{\gamma,a}|A^*(\delta_{\gamma ',b} )\right\rangle = \left\langle A (\delta_{\gamma,a})|\delta_{\gamma ',b} \right\rangle =\overline{[A(\gamma',\gamma)]_{b,a} },$$
%we have that $$A=A^*  \Leftrightarrow A(\gamma,\gamma') = A(\gamma',\gamma)^*, \quad \forall \gamma,\gamma'\in\Z^2.$$
%So $$\O_\varepsilon(A)=\O_\varepsilon(A)^*\Leftrightarrow e^{i\varepsilon\varphi(\gamma,\gamma')} A(\gamma,\gamma') = e^{-i\varepsilon\varphi(\gamma ',\gamma)} A(\gamma',\gamma)^*, \quad \forall \gamma,\gamma'\in\Z^2, $$
%and the skew-symmetry of $\varphi$ proves the fourth statement.

The only non-trivial property is the third one. For every $\gamma,\gamma',\gamma''\in\Z^2$ it holds that 
    \begin{align*}
        (\tau_{\vep,\gamma''} Af)(\gamma)&= e^{i\varepsilon\varphi(\gamma,\gamma'')} (Af(\gamma-\gamma''))=\sum_{\gamma'\in\Z^2}e^{i\vep\varphi(\gamma,\gamma'')} A(\gamma-\gamma'',\gamma')f(\gamma'),\\
        (A\tau_{\vep,\gamma''} f)(\gamma)&=\sum_{\gamma'\in\Z^2} A(\gamma,\gamma')(\tau_{\varepsilon,\gamma''}f(\gamma')) =\sum_{\gamma'\in\Z^2} A(\gamma,\gamma') e^{i\vep\varphi(\gamma',\gamma'')}f(\gamma'-\gamma'')\\
        &=\sum_{\gamma'\in\Z^2} A(\gamma,\gamma'+\gamma'') e^{i\vep\varphi(\gamma'+\gamma'',\gamma'')}f(\gamma').
    \end{align*}
    So, $A$ is $(\varepsilon,q)$-periodic if and only if its hopping matrices satisfy the following constraint for every $\gamma,\gamma' \in\Z^2$, $\gamma''\in q\Z\times \Z$:
    \begin{equation}\label{eq:condition_commuting_with_mag_translations}A(\gamma-\gamma'',\gamma')e^{i\vep\varphi(\gamma,\gamma'')} = A(\gamma,\gamma'+\gamma'')e^{i\varepsilon\varphi(\gamma',\gamma'')}\end{equation}
    An identical procedure ensures that $\O_\eta (A)$ is $(\varepsilon+\eta,q)$-periodic if and only if its hopping matrices satisfy the following constraint for every $\gamma,\gamma' \in\Z^2$, $\gamma''\in q\Z\times \Z$:
    \begin{equation*}[\O_\eta (A)](\gamma-\gamma'',\gamma')e^{i(\vep+\eta )\varphi(\gamma,\gamma'')} = [\O_\eta( A)](\gamma,\gamma'+\gamma'')e^{i(\varepsilon+\eta)\varphi(\gamma',\gamma'')}\end{equation*}
    Using the definition of $\O_\eta( A)$ this is equivalent to 
    \begin{equation}\label{eq:auxiliary_equation_11}
        e^{i\eta\varphi(\gamma-\gamma'',\gamma')} A(\gamma-\gamma'',\gamma')e^{i(\vep+\eta )\varphi(\gamma,\gamma'')} =e^{i\eta\varphi(\gamma,\gamma'+\gamma'')}A(\gamma,\gamma'+\gamma'')e^{i(\varepsilon+\eta)\varphi(\gamma',\gamma'')}.
    \end{equation}
    Since $\varphi$ is bilinear and anti-symmetric, Equation \eqref{eq:auxiliary_equation_11} is equivalent to Equation \eqref{eq:condition_commuting_with_mag_translations}, proving the final statement.
\end{proof}

Now we want to investigate how close $\O_\varepsilon$ is to a homomorphism property of the type $\O_\varepsilon(A\circ B) = \O_\varepsilon(A) \O_\varepsilon(B)$. The next Lemma states that the property holds up to local errors for local operators.

\begin{lemma}\label{lemmahc2}
Consider two operators $A,B\in\mathcal{B}^{\exp} (\mathcal{H})$. 
Define the operator $A\triangle B$ through its matching matrices:
    \begin{equation*}\label{eq:bilinear_superoperator_triangle_circ}
        (A\triangle B)(\gamma,\gamma')=i \sum_{\gamma'' \in \Z^2} \varphi(\gamma-\gamma'',\gamma''-\gamma')\, A(\gamma,\gamma'')B(\gamma'',\gamma')
    \end{equation*}
Then $A\triangle B \in \mathcal{B}^{\exp}(\mathcal{H})$ and
    \begin{equation}\label{eq:product_related_to_delta}
        \O_\vep(A)\O_\vep(B)=\O_\vep(AB)+\vep \O_\vep(A\triangle B)+\mathcal{O}^{\exp} (\vep^2).
    \end{equation}
\end{lemma}

\begin{proof}
    We have %The hopping matrices of the operators in Equation \eqref{eq:product_related_to_delta} are:
    \[\begin{split}
    [\O_\vep(A)\O_\vep(B)](\gamma,\gamma') & = \sum_{\gamma''\in\Z^2} e^{i\varepsilon\varphi(\gamma-\gamma',\gamma'')}A(\gamma,\gamma'')B(\gamma'',\gamma'), %\\
    %\O_\varepsilon(AB)(\gamma,\gamma') &= \sum_{\gamma''\in\Z^2} e^{i\varepsilon\varphi(\gamma,\gamma')}A(\gamma,\gamma'')B(\gamma'',\gamma'), \quad \mbox{and} \\
%\varepsilon\O_\varepsilon(A\triangle B)(\gamma,\gamma') &=  \varepsilon \sum_{\gamma'' \in \Z^2} ie^{i\varepsilon\varphi(\gamma,\gamma')} \varphi(\gamma-\gamma'',\gamma''-\gamma')\, A(\gamma,\gamma'')B(\gamma'',\gamma'). 
    \end{split}\]
    thus the hopping matrices of $L= \O_\varepsilon (A)\O_\varepsilon (B) - \O_\varepsilon (AB) - \varepsilon \O_\varepsilon (A\triangle B)$ are equal to:
    \[L(\gamma,\gamma')=e^{i\varepsilon\varphi(\gamma,\gamma')} \sum_{\gamma''\in \Z^2} \left[e^{i\varepsilon[\varphi(\gamma-\gamma',\gamma'')-\varphi(\gamma,\gamma')]}-1-i\varepsilon\varphi(\gamma-\gamma'',\gamma''-\gamma')   \right]A(\gamma,\gamma'')B(\gamma'',\gamma'). \]
    Since $\varphi$ is bilinear and skew-symmetric, $$\varphi(\gamma-\gamma',\gamma'')-\varphi(\gamma,\gamma') = \varphi(\gamma-\gamma'',\gamma''-\gamma').$$
 Using the Taylor polynomial and the Lagrange remainder, we can estimate 
    \[|e^{ix}-1-ix|^2\le|\cos(x)-1|^2+|\sin(x)-x|^2 \le |x|^4\]
    and if $A$ and $B$ are local operators with exponential decay rates $w_a\geq w_b>0$, together with $|\varphi(\gamma-\gamma'',\gamma''-\gamma')|
    \leq 2^{-1}|\gamma-\gamma''|\, |\gamma''-\gamma'|$ and the triangle inequality we have 
\begin{equation*}\tag*{\qedhere}\begin{aligned}\|L(\gamma,\gamma')\| &\le C\vep^2\sum_{\gamma'' \in \Z^2} |\gamma-\gamma''|^2\, |\gamma''-\gamma'|^2e^{-w_b[|\gamma-\gamma''|+|\gamma''-\gamma'|]} \le C\vep^2 e^{-w_b|\gamma-\gamma'|/4}.
    \end{aligned}
    \end{equation*}
 %   After a substitution $\gamma-\gamma'' = \gamma'''$ and a use of the triangle inequality we can write
  %  \[\begin{split} \|L(\gamma,\gamma')\| &\le C |\varepsilon|^2\sum_{\gamma''' \in \Z^2}|\varphi(\gamma''',\gamma-\gamma'-\gamma''')|^2 e^{[-a|\gamma'''|-b|\gamma-\gamma'-\gamma'''|]}\\
   % &\le  C |\varepsilon|^2\left[\sum_{\gamma''' \in \Z^2}|\varphi(\gamma''',\gamma-\gamma')|^2 e^{-(a-b)|\gamma'''|}\right]e^{[-b|\gamma-\gamma'|]}. \end{split}\]
   % However, $\varphi$ behaves polynomially in its arguments, so the series is always well defined and can be estimated using a polynomial $p(|\gamma-\gamma'|) \in \R_+$:
    %\[\|L(\gamma,\gamma')\|\le |\varepsilon|^2 C p(|\gamma-\gamma'|)e^{-b|\gamma-\gamma'|}.\]
    %Up to rescaling the constant $C$ and the exponential weight $b$, it is always possible to have $p(|\gamma-\gamma'|)e^{-b|\gamma-\gamma'|} \le De^{-d|\gamma-\gamma'|}$, obtaining the thesis.
\end{proof}

\begin{corollary}
    If $A,B\in \mathcal{B}^{\exp} (\mathcal{H})$ are $(b_0,q)$-periodic, so is $A\triangle B$ and if this occurs all the operators of Equation \eqref{eq:product_related_to_delta} are $(b_0+\varepsilon,q)$-periodic.
\end{corollary}

\begin{proof}
    Straightforward computation using Equation \eqref{eq:condition_commuting_with_mag_translations}.
\end{proof}

\section{The almost invariant projection} \label{sec:P_epsilon_construction}

In order to define $P_\varepsilon$ in the case of a local spectral island, we need to construct four sequences $\vep_k$, $P_{k,\varepsilon}$, $W_k$, and $\Pi_{k,\varepsilon}$ using an iterative  procedure. %that follows these intuitive steps: if $[H_{b_0+\varepsilon}, P_{k,\varepsilon}]=\mathcal{O}^{\exp}(\varepsilon^k)$ we can isolate the leading term $W_k$ and use it to modify $P_{k,\varepsilon}$ into $\Pi_{k+1,\varepsilon}$ with $[\Pi_{k+1,\varepsilon}, H_{b_0+\varepsilon}]=\mathcal{O}^{\exp}(\varepsilon^{k+1})$.
%However, $\Pi_{k+1,\varepsilon}$ will not be a projection and we need to put it in a Riesz formula to obtain $P_{k+1,\varepsilon}$, but this step can only be done for suitable $|\varepsilon|<\varepsilon_{k+1}$.

\begin{condition}\label{def:P_k_epsilon} There exists a non-increasing sequence $\vep_k>0$ such that 
    $P_{k,\varepsilon}$ satisfies the following properties for all $k\in\mathbb{N}$ and $0\leq |\varepsilon|\leq \vep_k$:
    \begin{enumerate}
        \item It is an orthogonal projection acting on $\mathcal{H}$
        \item It is local, $P_{k,\varepsilon}\in\mathcal{B}^{\exp}(\mathcal{H})$
        \item It is $(b_0+\varepsilon,q)$-periodic
        %\item For every $\gamma,\gamma'\in\Z^2$ the function $\varepsilon\mapsto P_{k,\varepsilon} (\gamma,\gamma')$ is {\color{blue} smooth}
        \item $[H_{b_0+\varepsilon},P_{k,\varepsilon}] = \mathcal{O}^{\exp} (\varepsilon^k)$
        \item $P_{k,\varepsilon}=P_{k-1,\varepsilon}+\mathcal{O}^{\exp}(\varepsilon^{k-1})$ for $k>1$.
    \end{enumerate}
\end{condition}

\begin{condition}\label{def:W_k}
    $W_k$ has the following properties for all $k\in\mathbb{N}$:
    \begin{enumerate}
        \item It is a $(b_0,q)$-periodic and local operator $W_k\in \mathcal{B}^{\exp}(\mathcal{H})$
        \item $[H_{b_0+\varepsilon},P_{k,\varepsilon}] = \varepsilon^k \O_\varepsilon(W_k) + \mathcal{O}^{\exp}(\varepsilon^{k+1})$
        \item $P_0W_kP_0\equiv 0 = (1-P_0)W_k(1-P_0)$.
    \end{enumerate}
\end{condition}

\begin{condition}\label{def:PI_k_epsilon}
    $\Pi_{k,\varepsilon}$ satisfies the following properties for all $k \in\mathbb{N}$ and $|\vep|\leq \varepsilon_{k-1}$ (with $\vep_0:=1$):
    \begin{enumerate}
        \item It is a local operator, $\Pi_{k,\varepsilon}\in\mathcal{B}^{\exp}(\mathcal{H})$ 
        \item It is $(b_0+\varepsilon,q)$-periodic
        %\item For every $\gamma,\gamma'\in\Z^2$ the function $\varepsilon\mapsto \Pi_{k,\varepsilon} (\gamma,\gamma')$ is real analytic
        \item $\Pi_{k,\varepsilon}^2-\Pi_{k,\varepsilon} = \mathcal{O}^{\exp}(\varepsilon^k)$
        \item $[H_{b_0+\varepsilon},\Pi_{k,\varepsilon}] = \mathcal{O}^{\exp}(\varepsilon^k)$
    \end{enumerate}
\end{condition}

Our first step in the construction of $P_\vep$ is the following:
\begin{proposition}\label{pro:Sequence_of_P_k_epsilon}
     Conditions \ref{def:P_k_epsilon}, \ref{def:W_k}, and \ref{def:PI_k_epsilon} can be simultaneously satisfied. 
     % Here $\varepsilon_k$ are {\color{blue} chosen such that for all $|\varepsilon|< \varepsilon_k$ we have} $\|\Pi_{k,\varepsilon }^2 -\Pi_{k,\varepsilon}\|<1/8$.
\end{proposition}

\begin{proof}

    \begin{itemize}
        \item \emph{Construction of} $\Pi_{1,\varepsilon}$ {\it and} $\vep_1$

        Let $\Pi_{1,\varepsilon} = \O_\varepsilon(P_0)$.  Lemma \ref{lem:P_0_local} and Lemma \ref{lem:properties_O_epsilon} ensure that $\Pi_{1,\varepsilon}\in\mathcal{B}^{\exp}(\mathcal{H})$ is local and $(b_0+\varepsilon,q)$-periodic.
        %Since $\O_\varepsilon$ acts as pointwise multiplication with $e^{i\vep \varphi(\gamma,\gamma')}$ for all $\gamma,\gamma'\in\mathbb{Z}^2$,  the maps $\varepsilon\mapsto\Pi_{1,\varepsilon} (\gamma,\gamma')$ are real analytic. 
        Moreover, Equation \ref{eq:product_related_to_delta} shows that $$\Pi_{1,\varepsilon}^2-\Pi_{1,\varepsilon} = \varepsilon\O_\varepsilon(P_0\triangle P_0) + \mathcal{O}^{\exp} (\varepsilon^2) = \mathcal{O}^{\exp}(\varepsilon).$$ 
        Since $[H_{b_0},P_0]=0$, Equation \ref{eq:product_related_to_delta} implies that  $$[H_{b_0+\varepsilon},\Pi_{1,\varepsilon}] = [\O_\varepsilon(H_{b_0}), \O_\varepsilon(P_0)] = \varepsilon \O_\varepsilon(H_{b_0}\triangle P_0 - P_0\triangle H_{b_0}) + \mathcal{O}^{\exp} (\varepsilon^2) = \mathcal{O}^{\exp} ( \varepsilon). $$
        Thus all the properties in  Condition \ref{def:PI_k_epsilon} are satisfied. Finally, we fix $0<\vep_1 \leq 1$ small enough such that $\Vert \Pi_{1,\varepsilon}^2-\Pi_{1,\varepsilon}\Vert <1/8$ for all $|\vep|\leq \vep_1$. 

        \item \emph{From }$\Pi_{k,\varepsilon}$ \emph{to}  $\vep_k$ {\it and} $P_{k,\varepsilon}$ 

    Assume that $\Pi_{k,\varepsilon}$ satisfies all the properties of Condition  \ref{def:PI_k_epsilon}. The property $\Pi_{k,\varepsilon}^2-\Pi_{k,\varepsilon} = \mathcal{O}^{\exp}(\varepsilon^k)$ together with the estimate of the operator norm in Lemma \ref{lem:locality_VS_continuity} implies that there exists a value $\varepsilon_k$ such that for all $|\varepsilon|< \varepsilon_k$ we have $\|\Pi_{k,\varepsilon }^2 -\Pi_{k,\varepsilon}\|<1/8$. 
        We may choose  $\varepsilon_k \leq  \varepsilon_{k-1}$. 
        We have that for all $|\varepsilon|< \varepsilon_k$, the spectrum of $\Pi_{k,\varepsilon}$ is contained within two disks with centers at $0$ and $1$ and a radius smaller than $1/4$. For these $\varepsilon$ define:
        \begin{equation}\label{eq:definition_P_k_epsilon}
            P_{k,\varepsilon} = \frac{1}{2\pi i}\oint_{|z-1|=1/2} (z-\Pi_{k,\vep})^{-1}dz.
        \end{equation}
        Corollary \ref{cor:locality_VS_functional_calculus} ensures that $P_{k,\varepsilon}\in\mathcal{B}^{\exp}(\mathcal{H})$ using the function $f\equiv 1$, and the functional calculus implies that $P_{k,\varepsilon}^2=P_{k,\varepsilon}$ and that $P_{k,\varepsilon}^* = P_{k,\varepsilon}$.
        $P_{k,\varepsilon}$ is $(b_0+\varepsilon,q)$-periodic as the argument inside the integral satisfies the same property. 
%        Essentially, $P_{k,\varepsilon}$ projects into the part of the spectrum of $\Pi_{k,\varepsilon} $ contained in the ball with center $1$ and radius $1/2$. 
Now we prove that 
        \begin{equation}\label{eq:relation_P_PI}
            P_{k,\varepsilon} = \Pi_{k,\varepsilon} + \mathcal{O}^{\exp} (\varepsilon^k).
        \end{equation}
        Let $\Delta_{k,\vep}\coloneq\Pi_{k,\vep}^2-\Pi_{k,\vep}$. We have 
         $$(z-\Pi_{k,\varepsilon})(\Pi_{k,\varepsilon}+z-1) = [z(z-1)]\left(1-\frac{\Delta_{k,\vep}}{z(z-1)} \right).$$
        If $|1-z|=1/2$ and $|\varepsilon|<\varepsilon_k$, we have $|z(z-1)| \geq 1/4$, so the last term is invertible since $\Vert \Delta_{k,\vep}\Vert <1/8$. 
        %$$\left(1-\frac{\Pi_{k,\varepsilon}(\Pi_{k,\varepsilon}-1)}{z(z-1)} \right)^{-1} = \sum_{l=0}^\infty \left(\frac{\Pi_{k,\varepsilon}(\Pi_{k,\varepsilon}-1)}{z(z-1)} \right)^l$$
        %Combining the last two Equations leads to $(z-\Pi_{k,\varepsilon})$ being invertible with:
        Thus introducing the function 
        $$h_z(\zeta)\coloneq\big (1-\tfrac{\zeta}{z(z-1)}\big )^{-1}-1,\quad |\zeta|<1/4,$$
        we have 
        \begin{align*}(z-\Pi_{k,\varepsilon})^{-1} &= \frac{ \Pi_{k,\varepsilon}+z-1}{z(z-1)} \left(1-\frac{\Delta_{k,\vep}}{z(z-1)} \right)^{-1}\\
        &=\frac{1}{z(z-1)}\Pi_{k,\varepsilon}+\frac{1}{z} +
        \frac{ \Pi_{k,\varepsilon}+z-1}{z(z-1)}h_z(\Delta_{k,\vep}).
        \end{align*}
We can factorize $h_z(\zeta)=\zeta f_z(\zeta)$ with $f_z$ analytic inside $|\zeta|<1/4$, hence by integrating in $z$ on $|z-1|=1/2$ we obtain 
\begin{equation}\label{oct11}
P_{k,\vep}=\Pi_{k,\vep}+\Delta_{k,\vep}\frac{1}{2\pi i}\oint_{|z-1|=1/2}\frac{ \Pi_{k,\varepsilon}+z-1}{z(z-1)}f_z(\Delta_{k,\vep}) dz.
\end{equation}
Corollary \ref{cor:locality_VS_functional_calculus} implies that $f_z(\Delta_{k,\vep})$ is local  uniformly in $z$, while the factor $\Delta_{k,\vep}$ is of order $\mathcal{O}^{\rm exp}(\vep^k)$, which implies the same property for $P_{k,\vep}-\Pi_{k,\vep}$.
        
       % Using the right hand side inside the integral of Equation \eqref{eq:definition_P_k_epsilon} proves the following equality:
        %$$P_{k,\varepsilon} = \sum_{l=0}^\infty   \frac{(\Pi_{k,\varepsilon}(\Pi_{k,\varepsilon}-1))^l}{2\pi i} \left[\int_{|z-1|=1/2}\frac{dz}{z^{l+1}(z-1)^l}+ \Pi_{k,\varepsilon}\int_{|z-1|=1/2}\frac{dz}{z^{l+1}(z-1)^{l+1}} \right]$$
        %The integrals are solvable using the Cauchy integral formula:
        %$$\frac{1}{2\pi i}\int_{|z-1|=1/2}\frac{z^{-j}dz}{(z-1)^{l+1}} = \frac{1}{l!} \left. \frac{\partial^l z^{-j}}{\partial z^l }\right|_{z=1 } $$
        %So, the previous series admits an explicit description of the first two terms
        %\begin{equation}\label{eq:development_of_Riesz_formula}P_{k,\varepsilon} = \Pi_{k,\varepsilon} +\Pi_{k,\varepsilon}(\Pi_{k,\varepsilon}-1) (1-2\Pi_{k,\varepsilon}) + \Pi_{k,\varepsilon}^2(\Pi_{k,\varepsilon}-1)^2 (\cdots).\end{equation}
        %We assumed that $\Pi_{k,\varepsilon}(\Pi_{k,\varepsilon}-1) = \mathcal{O}^{\exp}(\varepsilon^k)$, 
        This proves Equation \eqref{eq:relation_P_PI}.
        The combination of properties $[H_{b_0+\varepsilon},\Pi_{k,\varepsilon}] = \mathcal{O}^{\exp}(\varepsilon^k)$ and $\Pi_{k,\varepsilon} = P_{k-1,\varepsilon} + \mathcal{O}^{\exp}(\varepsilon^{k-1})$ together with Equation \eqref{eq:relation_P_PI} ensures that the two final properties $[H_{b_0+\varepsilon},P_{k,\varepsilon}] = \mathcal{O}^{\exp}(\varepsilon^k)$ and $P_{k,\varepsilon} = P_{k-1,\varepsilon} + \mathcal{O}^{\exp}(\varepsilon^{k-1})$ hold.

        \item \emph{From }$P_{k,\varepsilon}$ \emph{to} $W_{k}$

        Consider $P_{k,\varepsilon}$ satisfying all the properties of Definition \ref{def:P_k_epsilon}, then we define the operator $W_k $ as the leading order of $[H_{b_0+\varepsilon},P_{k,\varepsilon}] = \mathcal{O}^{\exp}(\varepsilon^k)$, which means 
        \begin{equation}\label{eq:auxiliary_equation}
        [H_{b_0+\varepsilon},P_{k,\varepsilon}] = \varepsilon^k \O_\varepsilon(W_k) + \mathcal{O}^{\exp}(\varepsilon^{k+1}).
        \end{equation}
        %, to be precise, its hopping matrices are:
        
        %This is well defined because $H_{b_0+\varepsilon}$ and $P_{k,\varepsilon}$ are real analytic around $\varepsilon=0$. The locality follows because \[\begin{split}[H_{b_0+\varepsilon},P_{k,\varepsilon}]=\mathcal{O}^{\exp}(\varepsilon^k) &\Rightarrow \|[H_{b_0+\varepsilon},P_{k,\varepsilon}](\gamma,\gamma')\| \le C |\varepsilon|^k  e^{-a|\gamma-\gamma'|} \\&\Rightarrow\|W_k(\gamma,\gamma')\| \le Ce^{-a|\gamma-\gamma'|}.\end{split}\]
        Since $[H_{b_0+\varepsilon},P_{k,\varepsilon}]$ is $(b_0+\varepsilon,q)$-periodic, Lemma \ref{lem:properties_O_epsilon} implies that $\O_\varepsilon^{-1}([H_{b_0+\varepsilon},P_{k,\varepsilon}])$ is $(b_0,q)$-periodic and its limit 
        \begin{equation*}\label{eq:W_k_definition}
           W_k (\gamma,\gamma') = \lim_{\varepsilon\to 0} \frac{\O_\varepsilon^{-1}([H_{b_0+\varepsilon},P_{k,\varepsilon}])(\gamma,\gamma')}{\varepsilon^k}
       \end{equation*}
       must also be $(b_0,q)$-periodic.
        
        The equality $P_{k,\varepsilon} = P_{k-1,\varepsilon} + \mathcal{O}^{\exp} ( \varepsilon^{k-1})$ leads to $P_{k,\varepsilon} = P_{1,\varepsilon} + \mathcal{O}^{\exp} ( \varepsilon)$, and Equation \eqref{eq:relation_P_PI} plus the properties at $k=1$ imply that  
        \begin{equation}\label{eq:auxiliary_equation_2} P_{k,\varepsilon} = \O_\varepsilon (P_0) + \mathcal{O}^{\exp} ( \varepsilon). \end{equation} 
        
        These estimates transform the equation
        $$P_{k,\varepsilon}[H_{b_0+\varepsilon},P_{k,\varepsilon}]P_{k,\varepsilon} \equiv0 \quad \mbox{into} $$$$ [\O_\varepsilon (P_0) + \mathcal{O}^{\exp} ( \varepsilon)] [\O_\varepsilon(W_k) + \mathcal{O}^{\exp}(\varepsilon)][\O_\varepsilon (P_0) + \mathcal{O}^{\exp} ( \varepsilon)]\equiv 0,$$
        further simplified into $$\O_\varepsilon (P_0) \O_\varepsilon(W_k)\O_\varepsilon (P_0) + \mathcal{O}^{\exp} ( \varepsilon)\equiv 0.$$
        Equation \eqref{eq:product_related_to_delta} implies $P_0 W_kP_0 + \mathcal{O}^{\exp} ( \varepsilon)\equiv 0$,
      for all $|\vep|$ small enough,  hence $P_0 W_kP_0= 0$. A similar argument leads to $(1-P_0) W_k(1-P_0)=0$. 

        \item \emph{From }$P_{k,\varepsilon}$ \emph{and} $W_{k}$ \emph{to} $\Pi_{k+1,\varepsilon}$ 

        We need to introduce another auxiliary object. \begin{definition}[Adjusting super operator $T$]\label{def:T}
    Consider $p_0(k)$, $g_j(k)$, and $\mathcal{C}_j$ being the same as those in the proof of Lemma \ref{lem:P_0_local}. For any $(b_0,q)$-periodic operator $W\in\mathcal{B}^{\exp}(\mathcal{H})$ with fibers $w(k)$, we define the operator $T(W)$ through its fibers:
    \begin{equation}\label{eq:T_definition}
    \begin{aligned}
    t(W)(k)\coloneq &\sum_{j=1}^L \frac{g_j(k)}{2\pi i}\times \oint_{\mathcal{C}_j} (z-h(k))^{-1}w(k) (z-h(k))^{-1} (1-p_{0}(k)) dz \\
   &
    -\sum_{j=1}^L \frac{g_j(k)}{2\pi i}\times\oint_{\mathcal{C}_j}(1-p_{0}(k))(z-h(k))^{-1}w(k) (z-h(k))^{-1}\, dz.
    \end{aligned}
    \end{equation}
\end{definition}

\begin{lemma}\label{lem:T_properties}
    The super operator $T$ has the following properties:
    \begin{enumerate}
        \item If $W$ is local and $(b_0,q)$-periodic, then $T(W)$ is local and $(b_0,q)$-periodic
        \item $P_0 T(W)P_0 =0 =  (1-P_0)T(W) (1-P_0) $
        \item $[H_{b_0},T(W)] = P_0W(1-P_0) + (1-P_0)WP_0$
    \end{enumerate}
\end{lemma}

\begin{proof}
    For the first part, we have that $T(W)$ is defined through its fibers, so it is $(b_0,q)$-periodic and local due to Lemma \ref{lem:local+periodic=analytic_fibers}.
    The second equality holds because $\forall k\in \mathbb{T}^2$, $p_0(k)$ is a spectral projection of $h(k)$, so $$p_0(k)^2 = p_0(k), \quad (1-p_0(k))^2 = (1-p_0(k)),\quad \mbox{and} \quad [h(k),p_0(k)]=0.$$
    For the third equality, we have that $h(k)$, $p(k)$, and $(z-h(k))^{-1}$ commute with each other, so $[H_{b_0},T(W)] = T([H_{b_0},W])$.
    Since the fibers of $[H_{b_0},W]$ are $[h(k),w(k)]$, changing $w(k)$ with $ h(k)w(k)-w(k)h(k)$ in Equation \eqref{eq:T_definition} leads to terms such as $h(k)(z-h(k))^{-1}$ and $(z-h(k))^{-1}h(k)$, which are equal to $-1+z(z-h(k))^{-1}$. 
    But, all terms related to $z(z-h(k))^{-1}$ simplify, obtaining that the fibers $T([H_{b_0},W])(k)$ are equal to:
    \[\begin{split}& \sum_{j=1}^L \frac{g_j(k)}{2\pi i} \times\oint_{\mathcal{C}_j}\big[ -w(k) (z-h(k))^{-1} (1-p_{0}(k)) + (z-h(k))^{-1} w(k) (1-p_{0}(k)) \big] dz \\ &
    -\sum_{j=1}^L \frac{g_j(k)}{2\pi i}\times\oint_{\mathcal{C}_j}\big[-(1-p_{0}(k))w(k) (z-h(k))^{-1} + (1-p_{0}(k)) (z-h(k))^{-1} w(k)\big]\, dz.\end{split}\]
    Now, the integrals and the sums can be reabsorbed using Equation \eqref{eq:fibers_of_P_0} to have that 
    \[\begin{split}t ( [ H_{b_0} ,W ] ) (k) = &-w(k)p_0(k) (1-p_0(k)) + p_0(k)w(k)(1-p_0(k))\\ & + (1-p_0(k))w(k) p_0(k) - (1-p_0(k)) p_0(k)w(k) .\end{split} \]
    In the end, the equality $p_0(k)(1-p_0(k))\equiv 0$ finishes the proof.
\end{proof}

        If $P_{k,\varepsilon}$ and $W_{k}$ satisfy all the properties from Conditions \ref{def:P_k_epsilon} and \ref{def:W_k}, then we define 
        \begin{equation}\label{eq:definition_PI_k_epsilon}
            \Pi_{k+1,\varepsilon} = P_{k,\varepsilon} -\varepsilon^k\O_\varepsilon(T(W_k)).
        \end{equation}
        It is a local operator since $T$ preserves locality (Lemma \ref{lem:T_properties}) and both operators are local. 
        
        Since $T$ preserves $(b_0,q)$-periodicity, Lemma \ref{lem:properties_O_epsilon} ensures that $\O_\varepsilon(T(W_k))$ is $(b_0+\varepsilon,q)$-periodic and by hypothesis, $P_{k,\varepsilon}$ does the same, so $\Pi_{k+1,\varepsilon}$ is also $(b_0+\varepsilon,q)$-periodic.
        As for the properties of the order, Lemma \ref{eq:product_related_to_delta}, the properties of Definitions \ref{def:P_k_epsilon} and Lemma \ref{lem:properties_O_epsilon} imply that
        \[\begin{split} [\O_\varepsilon \big (T(W_k)\big ),\O_\varepsilon (H_{b_0}) ] & = \O_\varepsilon (\big [T(W_k),H_{b_0}]\big ) + \mathcal{O}^{\exp}(\varepsilon) \\ & = \O_\varepsilon(P_0W_k(1-P_0)  + (1-P_0)W_kP_0)) + \mathcal{O}^{\exp}(\varepsilon) \\ &= \O_\varepsilon(W_k) +\mathcal{O}^{\exp}(\varepsilon).\end{split}\]
        This together with Equation \eqref{eq:auxiliary_equation} lead to:
        \[ \begin{split} [\Pi_{k+1,\varepsilon}, H_{b_0+\varepsilon}] &= [P_{k,\varepsilon}, H_{b_0+\varepsilon}] - \varepsilon^k [\O_\varepsilon \big (T(W_k)\big ),\O_\varepsilon (H_{b_0}) ] \\& = [P_{k,\varepsilon}, H_{b_0+\varepsilon}] - \varepsilon^k \O_\varepsilon(W_k) +\mathcal{O}^{\exp}(\varepsilon^{k+1}) = \mathcal{O}^{\exp}(\varepsilon^{k+1}) \end{split}\]
        Finally, consider
        \[\Pi_{k+1,\varepsilon}^2 = P_{k,\varepsilon}^2 -\varepsilon^k\O_\varepsilon(T(W_k)) P_{k,\varepsilon}  - P_{k,\varepsilon}\varepsilon^k\O_\varepsilon(T(W_k)) + \mathcal{O}^{\exp}(\varepsilon^{k+1}).\]
        Equation \eqref{eq:auxiliary_equation_2}, the second property of $T$ in Lemma \ref{lem:T_properties}, and Equation \eqref{eq:product_related_to_delta}, allow to write the middle terms as
        \[\begin{split}\O_\varepsilon(T(W_k)) P_{k,\varepsilon} + &P_{k,\varepsilon}\O_\varepsilon(T(W_k))  \\ &= \O_\varepsilon(T(W_k)) \O_\varepsilon (P_0) + \O_\varepsilon (P_0)\O_\varepsilon(T(W_k)) + \mathcal{O}^{\exp}(\varepsilon) \\ 
        &=\O_\varepsilon(T(W_k)P_0+P_0T(W_k) ) + \mathcal{O}^{\exp}(\varepsilon) \\ 
        &=\O_\varepsilon(T(W_k))+\mathcal{O}^{\exp}(\varepsilon).\end{split} \]
        Hence we have 
        \[\Pi_{k+1,\varepsilon}^2 = P_{k,\varepsilon}^2 - \varepsilon^k  (\O_\varepsilon(T(W_k))+\mathcal{O}^{\exp}(\varepsilon) ) +\mathcal{O}^{\exp}(\varepsilon^{k+1}) =\Pi_{k+1,\varepsilon} +\mathcal{O}^{\exp}(\varepsilon^{k+1}).\]
        Thus we obtained all the properties of $\Pi_{k+1,\varepsilon}$ and the proof of Proposition \ref{pro:Sequence_of_P_k_epsilon} is over.\qedhere
    \end{itemize}    
\end{proof}

%\begin{remark}[The role of $q$ in $b_0=\tfrac{2\pi p}{q}$]\label{rem:role_q}
 %   The relation between $q$ and $\varepsilon_1$ is chaotic. On one hand $P_0\triangle P_0$ is $(b_0,q)$-periodic, so $\mathcal{U}(P_0\triangle P_0) \mathcal{U}^{-1}$ admits a fiber decomposition obtained using Equations \eqref{eq:fiber_of_periodic_operator} and \eqref{eq:bilinear_superoperator_triangle_circ}:
  %  \begin{equation}
   %     \begin{split}
    %    (P_0\triangle P_0)(k)|_{jn+a,j'n+a'} = & \sum_{\alpha,\gamma\in\Z^2} e^{ik\cdot \alpha} e^{-i\pi\frac{p}{q}\alpha_2(j'+q\alpha_1) } i\varphi \big((j,0)-\gamma, \gamma-(j'+q\alpha_1,\alpha_2)\big) \\ &  \cdot \big[ P_0((j,0),\gamma) P_0(\gamma,(j'+q\alpha_1,\alpha_2))\big]\big|_{a,a'}
     %   \end{split}
    %\end{equation}
    %The $q$ present in the right argument of $\varphi$ implies that there is a linear dependence on $q$ and this suggests that $\| P_0\triangle P_0 \|$ may grow as $q$ increases.
    %On the other hand, the definition of $\varepsilon_1$ relies on the constraint $\|\O _\varepsilon(P_0)^2-\O_\varepsilon(P_0)\|<1/8$ which requires to control $\|\varepsilon \O_\varepsilon(P_0\triangle P_0)\|$. 
    %However, $\O_\varepsilon$ does not preserve the norm, so, even if $\| P_0\triangle P_0 \|$ grows when $q$ increase, it may happen that $\|\varepsilon O_\varepsilon(P_0\triangle P_0)\|$ does not, so we expect to have no explicit relation between $q$ and $\varepsilon_1$.
%\end{remark}

Now we are ready to construct the almost invariant projection $P_\varepsilon$ announced in Theorem \ref{thm:main_theorem} \rm{(ii)}. The idea is to put $P_\varepsilon = P_{k(\varepsilon),\varepsilon}$ for some function $k\colon(0,\varepsilon_1)\to \N$ with $k(\varepsilon)\to \infty$ as $\varepsilon\to 0$.
%To do so, we cannot directly use the sequence $P_{k,\varepsilon}$, but we will need an auxiliary sequence $\mathfrak{P}_{k,\vep}$. Then we will be able to define $P_\varepsilon $ as a specific element of the second sequence and its choice is linked to the convergence radius of a particular power series.

\begin{proposition}\label{pro:construction_P_epsilon}
     Consider the sequences $\{\varepsilon_k\}_{k\ge 1}$ and $P_{k,\varepsilon}$ obtained thanks to Proposition \ref{pro:Sequence_of_P_k_epsilon}. 
     Then, for all $|\varepsilon|< \varepsilon_1$, there is an orthogonal local projection $P_\vep=P_\vep^2=P_\vep^*\in \mathcal{B}^{\exp}(\mathcal{H})$ such that:
     \begin{itemize}
         \item It is $(b_0+\varepsilon,q)$-periodic
         \item $\lim_{\vep\rightarrow 0}\|P_\vep-\O_\vep(P_0)\|=0$
         \item For every $l\in\N$ it holds that $\lim_{\vep\rightarrow 0}\Big (|\vep|^{-l}\Vert [H_{b_0+\varepsilon},P_\vep]\Vert \Big ) =0$
     \end{itemize}
\end{proposition}

\begin{proof}

        The fact that $P_{k,\vep}-P_{k-1,\vep} = \mathcal{O}^{\exp} (\varepsilon^{k-1})$ does not imply that $\|P_{k,\vep}-P_{k-1,\vep}\| $ goes to zero as $k\to \infty$. 
        %In fact, Lemma \ref{lem:locality_VS_continuity} states that the norm of a local operator is bounded by a constant that depends on the exponential weight that in this case depends on $k$. 
        The best estimate we can give comes from Lemma \ref{lem:locality_VS_continuity} and is actually the following: given $k\geq 1$ we have
        $$\| P_{k,\vep}-P_{k-1,\vep}\| \leq  d_k\, |\vep|^{k-1},\quad  0<|\vep|<\vep_k,\quad P_{0,\vep}\coloneq P_0.$$
        Now, define
        $$\mathfrak{P}_{k,\vep}= \begin{cases}
        P_{1,\vep} & \mbox{if } \vep_{k}\leq|\vep|< \vep_1,\\
        P_{k,\vep}& \mbox{if } 0<|\vep|< \vep_{k}.
        \end{cases} $$
        This implies
        $$\mathfrak{P}_{k+1,\vep}-\mathfrak{P}_{k,\vep}= \begin{cases}
        0 & {\rm if}\; \vep_{k}\leq|\vep|< \vep_1,\\
        P_{1,\vep}-P_{k,\vep} & {\rm if}\; \vep_{k+1}\leq|\vep|< \vep_{k},\\
        P_{k+1,\vep}-P_{k,\vep} & {\rm if}\; 0<|\vep|< \vep_{k+1}.
        \end{cases} $$
        Up to adjusting the constants to take care of the case $\vep_{k+1}\leq|\vep|< \vep_{k}$, there exists $\tilde{c}_k>0$ such that 
        $$\Vert \mathfrak{P}_{k+1,\vep}-\mathfrak{P}_{k,\vep}\Vert \, \leq \tilde{c}_k\, |\vep|^k,\quad 0<|\vep| < \vep_1,\quad k\geq 1.$$
        Applying the triangle inequality involving $\mathfrak{P}_{2k+1,\vep}$ leads to the existence of $c_k>0$ such that 
        \begin{equation*}\label{eq:norm_estimate_auxiliary_projection}
            \Vert \mathfrak{P}_{2k+2,\vep}-\mathfrak{P}_{2k,\vep}\Vert \, \leq c_k\, |\vep|^{2k},\quad 0<|\vep| <\vep_1,\quad k\geq 1.
        \end{equation*}

        Let us first consider the case in which the power series $\sum_{k\geq 1} c_kz^k$ has a positive convergence radius $r>0$, then for $0<|\vep|<\min\{\vep_1,r\}$ the sequence $\{\mathfrak{P}_{2k,\vep}\}_{k\geq 1}$ is Cauchy and 
        $$\sum_{k\geq l} \big \Vert \mathfrak{P}_{2k+2,\vep}-\mathfrak{P}_{2k,\vep}\big \Vert \leq C\, |\vep|^l.$$
        For $0< |\vep| <\min\{\vep_1,r\}$ we set 
        $$P_\vep \coloneq \mathfrak{P}_{2\lfloor |\vep|^{-1}\rfloor,\vep} $$
        where $\lfloor x\rfloor$ denotes the largest integer smaller than $x$, and the first property of $P_\varepsilon$ follows from the properties of $P_{k,\varepsilon}$. 
        Then, for a fixed $l\geq 1$ and $\vep$ near zero, we have $\lfloor |\vep|^{-1}\rfloor>l$, so 
        \begin{equation}\label{oct2}
        \Vert P_\vep- \mathfrak{P}_{2l,\vep}\Vert\leq \sum_{k\geq 2l} \Vert\mathfrak{P}_{2k+2,\vep}-\mathfrak{P}_{2k,\vep}\Vert \leq C\, |\vep|^{l}. 
        \end{equation}
        Since $\mathfrak{P}_{2l,\varepsilon} = P_{2l, \varepsilon}$ for $\varepsilon$ close enough to zero, when $l=1$, Equation \eqref{eq:relation_P_PI} and the estimate of Lemma \ref{lem:locality_VS_continuity} ensure that
        $$\|P_\varepsilon-\O_\varepsilon(P_0)\| \le \|P_\varepsilon-P_{2,\varepsilon}\|+\|P_{2,\varepsilon}-\O_\varepsilon(P_0)\|\le C'|\varepsilon|,$$
        obtaining the second property of $P_\varepsilon$. 
        For the third property, for any $\lfloor |\vep|^{-1}\rfloor>l$ it holds that $[H_{b_0+\varepsilon}, \mathfrak{P}_{2l,\vep}] = \mathcal{O}^{\exp}(\varepsilon^{2l})$ and the estimate of Lemma \ref{lem:locality_VS_continuity} ensures that $\|[H_{b_0+\varepsilon}, \mathfrak{P}_{2l,\vep}]\| \le D|\varepsilon|^{2l}$. 
        Thus for a fixed $l\in\N$ and also using \eqref{oct2} we obtain
        $$\|[H_{b_0+\varepsilon}, P_\vep]\|\le \|[H_{b_0+\varepsilon}, \mathfrak{P}_{2l,\vep}]\|+\|[H_{b_0+\varepsilon}, P_\vep-\mathfrak{P}_{2l,\vep}]\|\le C'|\vep|^l,\quad \forall l\geq 1.$$

        Now let us consider the case in which the power series $\sum_{k\geq 1} c_k z^k$ has a convergence radius equal to zero. In this case, for every $0<|\vep|< \vep_1$ we set $$ N(\vep)= \begin{cases}
            1 & \mbox{if }\quad  c_k|\vep|^k > 2^{-k},\quad \forall k\geq 1, \\ \mbox{largest }K\in\N &\mbox{such that }c_k|\vep|^k \leq 2^{-k},\quad \forall\, 1\leq k\leq K.
        \end{cases} $$ 
        Note that $N(\vep)<\infty$ for all $\vep$, otherwise the power series would be dominated by the convergent series with elements $2^{-k}$. 
        Moreover, $N(\vep)$ is non-increasing in $|\varepsilon|$ and goes to infinity as $|\vep|\downarrow 0$.

        Now, set $$P_\vep= \mathfrak{P}_{2N(\vep)+2,\vep},$$ and the first property of $P_\varepsilon$ follows naturally.
        Fix $l\geq 1$. 
        There exists $0<\vep(l)<\vep_1$ such that $N(\vep(l))>l$, thus $N(\vep)\geq N(\vep(l))>l$ for all $0<|\vep|<\vep(l)$. Then for such $\vep$ it holds that 
        $$P_\vep-\mathfrak{P}_{2l,\vep}=\sum_{k=l}^{N(\vep)}\big (\mathfrak{P}_{2k+2,\vep}-\mathfrak{P}_{2k,\vep}\big )$$
        and 
        $$\Vert P_\vep-\mathfrak{P}_{2l,\vep}\Vert \leq \sum_{k=l}^{N(\vep)}|\vep|^k\, 2^{-k}\leq |\vep|^{l}.$$
        As done before, when $l=1$, this estimate proves the second property of $P_\varepsilon$. 
        Moreover, this shows that  
        $$\limsup_{\vep\rightarrow 0}\Big ( |\vep|^{-l}\Vert P_\vep-\mathfrak{P}_{2l,\vep}\Vert\Big )  \leq 1.$$
        Then 
        $$|\vep|^{-l}\Vert P_\vep-\mathfrak{P}_{l,\vep}\Vert\leq |\vep|^{-l}\Vert P_\vep-\mathfrak{P}_{2l,\vep}\Vert +|\vep|^{-l}\Vert \mathfrak{P}_{2l,\vep}-\mathfrak{P}_{l,\vep}\Vert.$$
        Since $\mathfrak{P}_{2l,\vep}-\mathfrak{P}_{l,\vep} = P_{2l,\vep}-P_{l,\vep} $ for $\varepsilon$ near zero, it holds by the triangle inequality that $\mathfrak{P}_{2l,\vep}-\mathfrak{P}_{l,\vep} =l\, \mathcal{O}^{\exp}(\varepsilon^l)$, and the estimate of Lemma \ref{lem:locality_VS_continuity} ensures that $\|\mathfrak{P}_{2l,\vep}-\mathfrak{P}_{l,\vep} \| \le C |\varepsilon|^l$, where the constant $C$ is  $l$-dependent. Therefore
        $$\limsup_{\vep\rightarrow 0}\Big (|\vep|^{-l}\Vert P_\vep-\mathfrak{P}_{l,\vep}\Vert\Big )  \leq 1 +\limsup_{\vep\rightarrow 0}\Big (|\vep|^{-l}\Vert \mathfrak{P}_{2l,\vep}-\mathfrak{P}_{l,\vep}\Vert\Big )\leq C_l<\infty.$$

        All the previous estimates imply that
        \begin{align*}
            \|[H_{b_0+\varepsilon}, P_\vep]\|\le \|[H_{b_0+\varepsilon}, \mathfrak{P}_{l,\vep}]\|+\|[H_{b_0+\varepsilon}, P_\vep-\mathfrak{P}_{l,\vep}]\| \le C_l'|\vep|^l,\quad \forall l\geq 1.\tag*{\qedhere}
        \end{align*}        
\end{proof}

\begin{remark}
    Although it could be desirable to have $$P_\varepsilon=\O_\varepsilon(P_0)+\mathcal{O}^{\exp}(\varepsilon) \quad  \mbox{and} \quad [H_{b_0+\varepsilon},P_\varepsilon]=\mathcal{O}^{\exp}(\varepsilon^l) \quad \forall l\in \N,$$ 
    with an exponential decay rate independent of $\vep$, we expect that this property does not hold true in the locally gapped case, in the sense that the exponential decay rate of $P_\vep$ might go to zero with $\vep$. The reason is that $P_\varepsilon$ is defined to be equal to $P_{k(\varepsilon),\varepsilon}$ for increasing values of $k(\varepsilon)$ when $\vep\to 0$. Since $P_{k+1,\varepsilon}$ is obtained from $P_{k,\varepsilon}$ after using a resolvent, Lemma \ref{lem:locality_of_resolvent} ensures that $P_{k+1,\varepsilon}$ is still local but with a smaller exponential decay rate $a(k)$. This means that at each step we lose a portion of the exponential decay rate and the proof of Lemma \ref{lem:locality_of_resolvent} suggests that $a(k) \to 0$ as $k \to \infty$. %This makes it extremely hard, if not impossible, to estimate the hopping matrices of $P_\varepsilon$ with $|\varepsilon|^l$.
\end{remark}

\subsection{The strictly gapped case}
\begin{lemma}\label{lem:P_epsilon_strictly_gapped}
    In the strictly gapped case, we can define $P_\vep$ to be a Riesz spectral projection of $H_{b_0+\vep}$,    satisfying the same properties as in Proposition \ref{pro:construction_P_epsilon} but with the optimal constraint $[H_{b_0+\varepsilon},P_\varepsilon]=0$. 
\end{lemma}

\begin{proof}
    We can define $P_0$ as we did in Equation \eqref{eq:strictly_gaped_zero_projection}, so that
    $$P_0=\frac{1}{2\pi i}\oint_{\mathcal{C}} (z-H_{b_0})^{-1} dz.$$
    For every $z\in \mathcal{C}$, Equation \eqref{eq:product_related_to_delta} and Lemma \ref{lem:locality_of_resolvent} imply that
    $$(z-H_{b_0+\varepsilon}) \O_\vep \big ((z-H_{b_0})^{-1}\big )=(z-\O_\vep(H_{b_0})) \O_\vep \big ((z-H_{b_0})^{-1}\big )=\Id +\mathcal{O}^{\exp}(\vep) \eqcolon \Id +A_\vep,$$
    where the error $A_\vep$ on the right hand side is local and of order $\vep$, uniformly in $z$ when $z\in \mathcal{C}$.
Let $f(x)=x/(1+x)$. Then if $|\vep|$ is small enough and using Corollary \ref{cor:locality_VS_functional_calculus} we have 
$$\big (\Id + A_\vep\big )^{-1}=\Id -f\big (A_\vep\big )=\Id + \mathcal{O}^{\exp} (\vep) .$$
    Hence for $\vep$ close enough to zero, $\mathcal{C}$ lies in the resolvent set of $H_{b_0+\varepsilon}$ and 
    $$(z-H_{b_0+\varepsilon})^{-1}=\O_\vep \big ((z-H_{b_0})^{-1}\big )+\mathcal{O}^{\exp}(\vep).$$
    Now we may define 
    $$P_\vep =\frac{1}{2\pi i}\oint_{\mathcal{C}} (z-H_{b_0+\varepsilon})^{-1} dz=\O_\vep(P_0) +\mathcal{O}^{\exp}(\vep)$$
    and all the desired properties hold due to functional calculus. This also shows that in the strictly gapped case, the rate of the exponential localization of $P_\vep$ is independent of $\vep$. 
\end{proof}

\section{Equality of the Chern numbers}\label{sec:Chern_numbers}

In this section we concentrate on the case of a local spectral gap, because the strictly gapped case was essentially treated in \cite{CorneanMonacoMoscolari2021}.
\begin{lemma}\label{lem:Ch_hopping_matrices_and_reduction}
    If $P\in \mathcal{B}^{\exp} (\mathcal{H})$ is a local projection, then $\forall \gamma\in\Z^2$ we have 
    \begin{equation} \label{eq:hopping_matrices_Chern_number}
        P\big[ [X_1,P],[X_2,P]\big] (\gamma,\gamma)= \sum_{\gamma',\gamma''\in\Z^2} 2\varphi(\gamma''-\gamma,\gamma'-\gamma'')P(\gamma ,\gamma')P(\gamma',\gamma'')P(\gamma'',\gamma).
    \end{equation}
    In addition, if $P$ is $(\varepsilon,q)$-periodic, we have
    \begin{equation}\label{oct6}
        {\rm Ch} (P) = \frac{2\pi i}{q} \sum_{j=0}^{q-1} \Tr\big ( P\big[[X_1,P],[X_2,P]\big] ((j,0),(j,0)) \big).
    \end{equation}
\end{lemma}

\begin{proof}
    Consider the multiplication operators $X_1,X_2$ such that $(X_j f) (\gamma) = \gamma_j f(\gamma) $. 
    %Their hopping matrices are
    %$$X_j(\gamma,\gamma') = \begin{cases}
 %       \gamma_j\Id & \mbox{if } \gamma=\gamma', \\
  %      0&\mbox{otherwise,}
   % \end{cases}%\quad X_2(\gamma,\gamma') = \begin{cases}
        %\gamma_2\Id & \mbox{if } \gamma=\gamma', \\
        %0&\mbox{otherwise.}
    %\end{cases}
   % $$
    Then
\begin{equation}\label{eq:hopping_matrices_commutators_with_X}[X_j,P](\gamma,\gamma') = (\gamma_j-\gamma_j')P(\gamma,\gamma').\end{equation}
    We now prove Equation \eqref{eq:hopping_matrices_Chern_number}:

    \[\begin{split}
        \begin{split}
            P\big[&[X_1,P],[X_2,P]\big](\gamma,\gamma)=\sum_{\gamma'\in\Z^2}P(\gamma,\gamma') \big[[X_1,P],[X_2,P]\big](\gamma',\gamma) \\
            &= \sum_{\gamma',\gamma'' \in\Z^2} P(\gamma,\gamma')\big[ [X_1,P](\gamma',\gamma'') [X_2,P](\gamma'',\gamma) -  [X_2,P](\gamma',\gamma'') [X_1,P](\gamma'',\gamma) \big ] \\
            &=\sum_{\gamma',\gamma''\in\Z^2} P(\gamma,\gamma') \big[ (\gamma_1'-\gamma_1'')P(\gamma',\gamma'')(\gamma_2''-\gamma_2) - (\gamma_2'-\gamma_2'')P(\gamma',\gamma'')(\gamma_1''-\gamma_1) \big]P(\gamma'',\gamma)  \\
            &= \sum_{\gamma',\gamma''\in\Z^2} 2\varphi(\gamma''-\gamma,\gamma'-\gamma'') P(\gamma,\gamma')P(\gamma',\gamma'')P(\gamma'',\gamma).
        \end{split}
    \end{split}\]

    Suppose that $P$ is $(\varepsilon,q)$-periodic and that $\gamma= (j+q\alpha_1,\alpha_2)$ with $j\in\{0,\dots,q-1\}$. Then Equation \eqref{eq:condition_commuting_with_mag_translations} ensures that
    \[\begin{split}
        P(\gamma,\gamma') &= P((j,0),\gamma '-(q\alpha_1,\alpha_2))e^{i\varepsilon\varphi(\gamma'-(j,0),-(q\alpha_1,\alpha_2))} \\
        P(\gamma',\gamma'')&=P(\gamma'-(q\alpha_1,\alpha_2),\gamma''-(q\alpha_1,\alpha_2)) e^{i\varepsilon\varphi (\gamma'-\gamma''+(q\alpha_1,\alpha_2),(q\alpha_1,\alpha_2))} \\
        P(\gamma'',\gamma)&=P(\gamma''-(q\alpha_1,\alpha_2),(j,0)) e^{i\varepsilon\varphi (\gamma''-(j,0), (q\alpha_1,\alpha_2))}.
    \end{split}\]
    Imposing these equalities in Equation \eqref{eq:hopping_matrices_commutators_with_X} and renaming $\gamma'-(q\alpha_1,\alpha_2)=\gamma'''$ and $\gamma''-(q\alpha_1,\alpha_2)=\gamma''''$ yields
    \[\begin{split}P&\big[ [X_1,P], [X_2,P]\big] (\gamma,\gamma)\\  &=\sum_{\gamma'''',\gamma''''\in\Z^2} 2\varphi(\gamma''''-(j,0),\gamma'''-(j,0))P((j,0) ,\gamma''')P(\gamma''',\gamma'''')P(\gamma'''',(j,0)) \\ &=P\big[ [X_1,P], [X_2,P]\big] ((j,0),(j,0)).\end{split}\]
    This allows the following simplification
    \begin{align*}
        \mathrm{Ch}(P) &= 2\pi i \lim_{L\to \infty} \frac{1}{(2L)^2} \sum_{\gamma\in [-L,L]^2}  \Tr\big ( P\big[[X_1,P],[X_2,P]\big] (\gamma,\gamma) \big) \\
        &=2\pi i \lim_{L\to \infty} \frac{1}{q4L^2}  \sum_{\alpha\in [-L,L]^2} \sum_{j=0}^{q-1}  \Tr\big ( P\big[[X_1,P],[X_2,P]\big] ((j+q\alpha_1,\alpha_2),(j+q\alpha_1,\alpha_2)) \big) \\
        &=\frac{2\pi i}{q} \sum_{j=0}^{q-1}  \Tr\big ( P\big[[X_1,P],[X_2,P]\big] ((j,0),(j,0)) \big).\tag*{\qedhere}
    \end{align*}
\end{proof}

Before we get additional results on the Chern number, we need an auxiliary result involving the relation between $P_{k,\vep}$ and $P_{k,\vep+\eta}$.

\begin{lemma}\label{lem:Order_P_PI_in_epsilon}
    In the framework of Proposition \ref{pro:Sequence_of_P_k_epsilon}, for any $k\in\N$, $0\le|\varepsilon|<\varepsilon_k$, and every $\eta\in\R$  such that $|\varepsilon+\eta| < \varepsilon_k$ we have:
    \begin{equation*}
        P_{k,\vep+\eta} = \O_\eta(P_{k,\vep}) + \mathcal{O}^{\exp}(\eta) \quad \mbox{and} \quad \Pi_{k,\vep+\eta} =\O_\eta( \Pi_{k,\vep}) + \mathcal{O}^{\exp}(\eta).
    \end{equation*}
\end{lemma}

\begin{proof} We proceed iteratively in $k$. 

    \begin{itemize}
        \item 
        {\it The starting point}:

        By definition, $\Pi_{1,\vep+\eta} = \O_{\varepsilon+\eta} (P_0)=\O_\eta(\O_\varepsilon(P_0))=\O_\eta(\Pi_{1,\vep})$. Thus here the remainder is absent. 
        
        \item \emph{From }$\Pi_{k,\vep+\eta}$ \emph{to} $P_{k,\vep+\eta}$

        Suppose that for $k\geq 1$ we have $\Pi_{k,\vep+\eta} =\O_\eta( \Pi_{k,\vep}) + \mathcal{O}^{\exp}(\eta)$.  Equation \ref{eq:product_related_to_delta} and Lemma \ref{lem:locality_of_resolvent} imply that:
        \[\begin{split}(z-\Pi_{k,\vep+\eta})\O_\eta \left[(z-\Pi_{k,\vep})^{-1} \right]&=   \left[\O_\eta(z-\Pi_{k,\vep}) + \mathcal{O}^{\exp} (\eta)\right] \O_\eta \left[(z-\Pi_{k,\vep})^{-1} \right] \\& = \Id + \mathcal{O}^{\exp} (\eta) \eqcolon\Id +A_\eta,\end{split}\]
        where $A_\eta$ is local,  of order $\eta$  and uniformly bounded in $z$ when $|z-1|=1/2$.
Let $f(x)=x/(1+x)$. Then if $|\eta|$ is small enough and using Corollary \ref{cor:locality_VS_functional_calculus} we have 
$$\big (\Id + A_\eta\big )^{-1}=\Id -f\big (A_\eta\big )=\Id + \mathcal{O}^{\exp} (\eta) .$$
        Therefore, it holds that:
        \[\begin{split}
            (z-\Pi_{k,\vep+\eta})^{-1} = \O_\eta \left[(z-\Pi_{k,\vep})^{-1}  \right] (\Id + \mathcal{O}^{\exp} (\eta)) =\O_\eta \left[(z-\Pi_{k,\vep})^{-1} \right]+\mathcal{O}^{\exp} (\eta).
        \end{split}\]
        Using the right hand side inside the integral of Equation \eqref{eq:definition_P_k_epsilon} ensures that
        \begin{equation}\label{oct3}
        P_{k,\vep+\eta} = \O_\eta(P_{k,\vep}) +\mathcal{O}^{\exp} (\eta)  \end{equation}
        due to the boundedness of the error in $z$ and the linearity of $\O_\eta$.

        \item \emph{From }$P_{k,\vep+\eta}$ \emph{to} $\Pi_{k+1,\vep+\eta}$

        Equations \eqref{eq:definition_PI_k_epsilon} and \eqref{oct3} imply that
        \[\begin{split}
            \Pi_{k+1,\varepsilon+\eta} & - \O_\eta(\Pi_{k+1,\varepsilon})  \\ &= P_{k,\varepsilon+\eta} - \O_\eta(P_{k,\varepsilon})- (\varepsilon+\eta)^k\O_{\varepsilon+\eta}(T(W_k)) + \O_\eta[\varepsilon^k\O_{\varepsilon}(T(W_k))] \\
            &= \mathcal{O}^{\exp}(\eta ) - \left[   (\varepsilon+\eta)^k-\varepsilon^k\right]\O_{\varepsilon+\eta}(T(W_k)) \\
            &= \mathcal{O}^{\exp}(\eta )
        \end{split}\]
        because $T(W_k)$ is local and $(\varepsilon+\eta)^k-\varepsilon^k$ is of order $\eta$. \qedhere
    \end{itemize}
\end{proof}

\begin{proposition}\label{pro:Ch_constant}
    Consider $P_{k,\varepsilon}$ as in Proposition \ref{pro:Sequence_of_P_k_epsilon} or $P_\varepsilon$ as in Proposition \ref{pro:construction_P_epsilon}. Then ${\rm Ch}(P_{k,\vep})={\rm Ch}(P_{0})$ for every $k\in \N$. Consequently, ${\rm Ch}(P_\vep)={\rm Ch}(P_{0})$.
\end{proposition}

\begin{proof}
    First, if $A\in\mathcal{B}^{\exp}(\mathcal{H})$ is local, then $[A,X_j]\in\mathcal{B}^{\exp}(\mathcal{H})$ is local:
    due to Equation \eqref{eq:hopping_matrices_commutators_with_X}, the hopping matrices are $[A,X_j](\gamma,\gamma') = (\gamma_j'-\gamma_j)A(\gamma,\gamma')$. If $\|A(\gamma,\gamma')\| \le Ce^{-a|\gamma-\gamma'|}$, then up to rescaling to two different constants, we have $$\|[A,X_j](\gamma,\gamma')\| \le C |\gamma_j'-\gamma_j| e^{-a|\gamma-\gamma'|} \le C' e^{-a'|\gamma-\gamma'|}.$$
    Similarly, for every $k\in\N$ and $0\le |\varepsilon+\eta| < \varepsilon_k$ we have $$\left([X_j,P_{k,\varepsilon+\eta}]-\O_\eta ( [X_j,P_{k,\varepsilon}]\right) (\gamma,\gamma') = (\gamma_j-\gamma_j')\left[\left( P_{k,\varepsilon+\eta}- \O_\eta ( P_{k,\varepsilon})\right) (\gamma,\gamma')\right] $$
    Now, Lemma \ref{lem:Order_P_PI_in_epsilon} states that $$ \left\|\left(P_{k,\varepsilon+\eta}- \O_\eta ( P_{k,\varepsilon})\right) (\gamma,\gamma')\right\|  \le C|\eta| e^{-a|\gamma-\gamma'|}.$$
    So, up to rescaling the constants we get
    $$\left\|\left([X_j,P_{k,\varepsilon+\eta}]-\O_\eta ( [X_j,P_{k,\varepsilon}]\right) (\gamma,\gamma')\right\| \le C'|\eta| e^{-a'|\gamma-\gamma'|}$$
    meaning that $$[X_j,P_{k,\varepsilon+\eta}]=\O_\eta ( [X_j,P_{k,\varepsilon}]) + \mathcal{O}^{\exp}(\eta). $$
    This equality allows the use of Equation \eqref{eq:product_related_to_delta}, implying that $$P_{k,\vep+\eta} [\, [X_1,P_{k,\vep+\eta}]\, , \, [X_2,P_{k,\vep+\eta}]\,] = \O_\eta\left( P_{k,\vep} [\, [X_1,P_{k,\vep}]\, , \, [X_2,P_{k,\vep}]\,] \right) + \mathcal{O}^{\exp}(\eta).$$
    Since $\O_\eta(A)((j,0)(j,0))=A((j,0),(j,0))$, Equation \eqref{oct6} implies that 
    \begin{equation*}\label{eq:Ch_behaviour}{\rm Ch}(P_{k,\vep+\eta})= {\rm Ch} (P_{k,\varepsilon}) + \frac{2\pi i}{q} \sum_{j=0}^{q-1} \Tr \big(\mathcal{O}^{\exp}(\eta) ((j,0),(j,0)) \big),\end{equation*}
    so the continuity in $\varepsilon$ follows for every $0\le |\varepsilon| < \varepsilon_k$. However, the Chern number is always integer-valued, so it must be constant for $\varepsilon$ in a neighborhood of zero. The final ingredient is $P_{k,0}=P_0$. 
\end{proof}

\section{The St\v{r}eda formula}\label{sec:Streda_formula}

%First, we prove that for $(b_0+\varepsilon,q)$-periodic operators, the integrated density of states assumes a simple formula.
%Then, we prove that two projections intertwined by a local and $(b_0+\varepsilon,q)$-periodic unitary equivalence have the same integrated density of states.
%Finally, we will focus on $P_{\varepsilon+\eta}$ and find a unitarily equivalent projection $Q_\eta$ whose expansion in $\eta$ of the integrated density of states has the Chern number ${\rm Ch}$ as its linear coefficient.

%\begin{remark}%\label{rmk:IDS_is_just_trace}
    If an operator $A$ is $(b_0+\varepsilon,q)$-periodic, then imposing $\gamma''=\gamma$ and $\gamma'=0$ in Equation \eqref{eq:condition_commuting_with_mag_translations}, yields $A(0,0)\equiv A(\gamma,\gamma)$ for all $\gamma \in q\Z\times\Z.$
    For such operators, the  integrated density of states exists and can be written as a ``trace per unit surface": 
    \begin{equation}\label{oct7}
    \mathcal{I}(A) = \lim_{L\to \infty} \frac{1}{(2L)^2} \sum_{\gamma\in [-L,L]^2} \Tr (A(\gamma,\gamma)) = \frac{1}{q}\sum_{j=0}^{q-1} \Tr(A((j,0),(j,0))) .
    \end{equation}
%\end{remark}

\begin{lemma}\label{lem:IDS_equal_for-un_eq_projections}
   Consider two local operators $A,B\in \mathcal{B}^{\exp}(\mathcal{H})$ that admit an integrated density of states.
    If there is a unitary $U\in \mathcal{B}^{\exp}(\mathcal{H})$ with $A=UBU^*$, then $\mathcal{I} (A)=\mathcal{I}(B)$.
\end{lemma}

\begin{proof}
    Define the function $f(\gamma,\gamma') = \Tr \big (U(\gamma,\gamma')[BU^*](\gamma',\gamma)\big )$. Because all operators are local, there exist $C,\alpha>0$ such that $|f(\gamma,\gamma')|\leq Ce^{-\alpha|\gamma-\gamma'|}$ for all $\gamma,\gamma'$. Then 
    \begin{equation}\label{oct4}
    \sum_{\gamma'\in\Z^2} f(\gamma,\gamma') =\Tr \left(\sum_{\gamma'\in\Z^2}U(\gamma,\gamma')[BU^*](\gamma',\gamma)\right) = \Tr(UBU^*(\gamma,\gamma)) = \Tr (A(\gamma,\gamma)), 
    \end{equation}
    and also 
    \begin{equation}\label{oct5}
    \sum_{\gamma\in\Z^2} f(\gamma,\gamma') =\Tr \left(\sum_{\gamma\in\Z^2}[BU^*](\gamma',\gamma)U(\gamma,\gamma')\right) = \Tr(BU^*U(\gamma',\gamma')) = \Tr (B(\gamma',\gamma')). 
    \end{equation}
    
    Now consider the quantity 
     $$I_L\coloneq\frac{1}{(2L)^2}\sum_{\gamma,\gamma'\in [-L,L]^2} f(\gamma,\gamma').$$
     Due to the exponential decay of $f$, the sum $\sum_{\gamma'\in [-L,L]^2} f(\gamma,\gamma')$ is uniformly bounded in $\gamma$ and $L$, and for a fixed $0<\delta<1$ we have 
     $$I_L=\frac{1}{(2L)^2}\sum_{\gamma\in [-L+L^\delta,L- L^\delta]^2}\quad \sum_{\gamma'\in [-L,L]^2} f(\gamma,\gamma')+\mathcal{O}(L^{-1+\delta}).$$
     Due to the same exponential localization, when $\gamma \in [-L+L^\delta,L- L^\delta]^2$ we may extend the summation in $\gamma'$ over the whole $\Z^2$, using the estimate 
     $$e^{-\alpha|\gamma-\gamma'|}\leq e^{-\alpha \, L^\delta } e^{-\alpha\,  {\rm dist}(\gamma',[-L,L]^2)},\quad \gamma'\in \Z^2\setminus [-L,L]^2,$$
which together with \eqref{oct4} leads to 
$$I_L=\frac{1}{(2L)^2}\sum_{\gamma\in [-L+L^\delta,L- L^\delta]^2}\Tr (A(\gamma,\gamma))+\mathcal{O}(L^{-1+\delta})\stackrel{L\rightarrow\infty}\rightarrow\mathcal{I}(A).$$
     
%     $$\lim_{L\to \infty}\frac{1}{(2L)^2}\sum_{\gamma,\gamma'\in [-L,L]^2} f(\gamma,\gamma') =  \lim_{L\to \infty}\frac{1}{(2L)^2}\sum_{\gamma'\in [-L,L]^2 }  \Tr(B(\gamma'\gamma'))= \mathcal{I}(B), $$
 %   but also
  %  \begin{align*}
   %     \lim_{L\to \infty}\frac{1}{(2L)^2}\sum_{\gamma,\gamma'\in [-L,L]^2} f(\gamma,\gamma') =  \lim_{L\to \infty}\frac{1}{(2L)^2}\sum_{\gamma\in [-L,L]^2 }  \Tr(A(\gamma,\gamma))= \mathcal{I}(A). \tag*{\qedhere}
   % \end{align*}
   In a similar way, using \eqref{oct5} and extending the summation of $\gamma$ in $I_L$ to the whole $\Z^2$, we obtain
   $$I_L\stackrel{L\rightarrow\infty}\rightarrow\mathcal{I}(B),$$
   hence $\mathcal{I}(A)=\mathcal{I}(B)$.
\end{proof}

\begin{proposition}\label{pro:derivative_of_IDS_of_P_epsilon}
    Let $P_\varepsilon$ be as in Proposition \ref{pro:construction_P_epsilon}. Then $$\mathcal{I}(P_{\varepsilon})=\mathcal{I}(P_{0}) +\frac{\vep\, {\rm Ch}(P_0)}{2\pi}.$$
\end{proposition}

\begin{proof} 
Since $P_\vep$ equals some $P_{k,\vep}$ for an $\vep$-dependent $k$, it suffices to prove 
\begin{equation}\label{oct9}
\mathcal{I}(P_{k,\varepsilon})=\mathcal{I}(P_{0}) +\frac{\vep\, {\rm Ch}(P_0)}{2\pi},\quad 0\leq |\vep|<\vep_k,\quad \forall k\geq 1.
\end{equation}
From Lemma \ref{lem:Order_P_PI_in_epsilon} we infer that for all $\gamma\in \Z^2$ we have 
$$P_{k,\vep+\eta}(\gamma,\gamma)=P_{k,\vep}(\gamma,\gamma)+\mathcal{O}(\eta)$$
uniformly in $\gamma\in \Z^2$. This implies that $\mathcal{I}(P_{k,\vep})$ is (Lipschitz) continuous in $\vep\in (-\vep_k,\vep_k)$ and $\mathcal{I}(P_{k,0})=\mathcal{I}(P_{0})$. Since we already know from Lemma \ref{pro:Ch_constant} that ${\rm Ch}(P_{k,\vep})={\rm Ch}(P_0)$, the identity \eqref{oct9} follows if we can prove that 
\begin{equation}\label{oct10}
\frac{d \mathcal{I}(P_{k,\vep})}{d\vep}=\frac{{\rm Ch}(P_{k,\vep})}{2\pi},\quad \vep\in (-\vep_k,\vep_k).
\end{equation}
The rest of the proof of this proposition is dedicated to \eqref{oct10}. 

        Fix $\varepsilon\in (-\varepsilon_k,\varepsilon_k)$ and let $\eta$ be near zero. The operator $\O_\eta (P_{k,\vep})$ is close to a projection because 

        $$\Delta_\eta\coloneq\big (\O_\eta (P_{k,\vep})\big )^2-\O_\eta (P_{k,\vep})= \mathcal{O}^{\rm exp}(\eta) $$
        and we may define $$Q_\eta=\frac{1}{2\pi i}\oint_{|z-1|=1/2}\big (z- \O_\eta (P_{k,\vep})\big )^{-1}\, dz.$$ 
        This is a $(\varepsilon+\eta,q)$-periodic and local projection. We may reuse the formula \eqref{oct11} where we put $Q_\eta$ instead of $P_{k,\vep}$ and $\O_\eta (P_{k,\vep})$ instead of $\Pi_{k,\vep}$, then expand it up to a second order in $\Delta_\eta$. This gives 
        \begin{equation}\label{oct13}
        Q_\eta = \O_\eta (P_{k,\vep})+\Delta_\eta -2\Delta_\eta\O_\eta (P_{k,\vep})+    \mathcal{O}^{\exp}(\eta^2).
        \end{equation}
        This also proves that $Q_\eta = \O_\eta (P_{k,\vep})+\mathcal{O}^{\exp}(\eta)$, which coupled to Lemma \ref{lem:Order_P_PI_in_epsilon} implies $Q_\eta = P_{k,\vep+\eta}+\mathcal{O}^{\exp}(\eta)$. Now, Lemma \ref{lem:locality_VS_continuity} ensures that $$\lim_{\eta \to 0}\left\|Q_\eta - P_{k,\vep+\eta}\right\|=0.$$
         For $\eta$ near zero we  define the Kato-Nagy intertwining unitary $$U_\eta =  \Big (P_{k,\varepsilon+\eta} Q_\eta + (\Id -P_{k,\varepsilon+\eta} ) (\Id -Q_\eta )  \Big ) \Big (\Id -\left(Q_\eta-P_{k,\varepsilon+\eta}\right)^2\Big )^{-1/2}.$$
         
        The triple $Q_\eta$, $P_{k,\vep+\eta}$ and $U_\eta$ satisfies the hypotheses of Lemma \ref{lem:IDS_equal_for-un_eq_projections}, meaning that $\mathcal{I}(P_{k,\varepsilon+\eta})=\mathcal{I}\left(Q_\eta\right)$.
        
        %For a general local operator $A$ it holds that \[\begin{split} \O_\eta (A)^2& (\gamma,\gamma)= \sum_{\gamma'\in\Z^2 } e^{i\eta [ \varphi (\gamma,\gamma') +\varphi(\gamma',\gamma)]} A(\gamma,\gamma')A(\gamma',\gamma) = A^2 (\gamma,\gamma) ,\quad \mbox{and}\\ 
        %\O_\eta (A)^3& (\gamma,\gamma)= \sum_{\gamma',\gamma''\in\Z^2 } e^{i\eta [ \varphi (\gamma,\gamma') +\varphi(\gamma',\gamma'')+\varphi(\gamma'',\gamma)]} A(\gamma,\gamma')A(\gamma',\gamma'') A(\gamma'',\gamma) \\
        %&= \sum_{\gamma',\gamma''\in\Z^2 } e^{-i\eta [ \varphi (\gamma''-\gamma,\gamma'-\gamma'')]} A(\gamma,\gamma')A(\gamma',\gamma'') A(\gamma'',\gamma) \\
        %&= \sum_{\gamma',\gamma''\in\Z^2 } \sum_{l=0}^\infty (-i\eta \varphi (\gamma''-\gamma,\gamma'-\gamma'') )^lA(\gamma,\gamma')A(\gamma',\gamma'') A(\gamma'',\gamma) \\
        %& = A^3 (\gamma,\gamma) -\sum_{\gamma',\gamma''\in\Z^2 } i\eta \varphi (\gamma''-\gamma,\gamma'-\gamma'') A(\gamma,\gamma')A(\gamma',\gamma'') A(\gamma'',\gamma)+\mathcal{O}(\eta^2).\end{split} \]
        Using \eqref{oct13} we may compute 
        \[\begin{split}Q_\eta (\gamma,\gamma) %&= 3 [\O_\eta(P_{k,\varepsilon})^2] (\gamma,\gamma) -2[\O_\eta (P_{k,\varepsilon})^3] (\gamma,\gamma ) + \mathcal{O}(\eta^2) \\
        &= P_{k,\varepsilon}(\gamma,\gamma) +\sum_{\gamma',\gamma''\in\Z^2 } 2i\eta \varphi (\gamma''-\gamma,\gamma'-\gamma'') P_{k,\varepsilon}(\gamma,\gamma')P_{k,\varepsilon}(\gamma',\gamma'') P_{k,\varepsilon}(\gamma'',\gamma) +\mathcal{O}(\eta^2) \\ &=P_{k,\varepsilon} (\gamma,\gamma) +i\eta \big ( P_{k,\varepsilon}\big[[X_1,P_{k,\varepsilon}],[X_2,P_{k,\varepsilon}]\big]\big) (\gamma,\gamma) + \mathcal{O}(\eta^2).\end{split}\]
Using \eqref{oct6} and \eqref{oct7}, the above identity yields 
$$\mathcal{I}(P_{k,\varepsilon+\eta})=\mathcal{I}(Q_\eta)=\mathcal{I}(P_{k,\vep}) +\frac{\eta}{2\pi} {\rm Ch}(P_{k,\vep}) +\mathcal{O}(\eta^2),$$
which proves \eqref{oct10}.\qedhere
%        Remark \ref{rmk:IDS_is_just_trace} provides a way to compute the integrated density of states and Lemma \ref{lem:Ch_hopping_matrices_and_reduction} relates it with the Chern number

 %       \[\begin{split}
  %          \left.\frac{d\mathcal{I}(P_{k,\varepsilon+\eta})}{d \eta}\right|_{\eta=0} & = \left.\frac{d\mathcal{I}(Q_{\eta})}{d \eta}\right|_{\eta=0} = \left.\frac{d}{d \eta}\frac{1}{q}\sum_{j=0}^{q-1} \Tr (Q_\eta ((j,0)(j,0) )\right|_{\eta=0} \\
            %&=\frac{i}{q}\sum_{j=0}^{q-1}\Tr \big ( P_{k,\varepsilon}\big[[X_1,P_{k,\varepsilon}],[X_2,P_{k,\varepsilon}]\big]((j,0),(j,0))\big) \\
            %&= \frac{1}{2\pi}{\rm Ch} (P_{k,\varepsilon}) \equiv \frac{1}{2\pi}{\rm Ch} (P_0).
        %\end{split}\]
        %This implies that $$\mathcal{I} (P_{k,\varepsilon}) = \mathcal{I}(P_0) +\varepsilon \frac{\rm{Ch} (P_0)}{2\pi},$$
        %since $P_\varepsilon=P_{k,\varepsilon}$ for some $k$, the last equation also prove the theorem for $P_\varepsilon$.
\end{proof}

%\begin{corollary}\label{cor:P_epsilon_discontinuous}
 %  Consider $P_{k,\varepsilon}$ as in Proposition \ref{pro:Sequence_of_P_k_epsilon}. If $\mathrm{Ch} (P_0)\ne 0$, then $P_{k,\varepsilon}$ is  nowhere continuous with respect to the operator norm. %\textcolor{red}{G.P. do you want also some comment on the weak norm?}
%\end{corollary}

%\begin{proof}
 %   Suppose that for a value $\varepsilon$ it holds that $$\lim_{\eta\to 0} \|P_{k,\varepsilon+\eta}-P_{k,\varepsilon}\| =0.$$ Then for $\eta$ near zero it is possible to construct a family of Kato-Nagy local unitaries $U_\eta$ between $P_{k,\varepsilon}$ and $P_{k,\varepsilon+\eta}$. Hence for $\eta$ near zero, Lemma \ref{lem:IDS_equal_for-un_eq_projections} implies that $\mathcal{I}(P_{k,\varepsilon}) = \mathcal{I}(P_{k,\varepsilon+\eta})$. This contradicts the fact that the slope of $\vep\mapsto \mathcal{I}(P_{k,\varepsilon})$ equals ${\rm Ch}(P_0)\neq 0.$ 
    %Therefore, $${\rm Ch} (P_0)=\left.\frac{d}{d\eta} \mathcal{I}(P_{k,\varepsilon+\eta}) \right|_{\eta=0}=0, $$
    %that contradicts the hypothesis. 
    %Since $P_\varepsilon $ is locally equal to some $P_{k,\varepsilon}$, the same argument holds also for $P_\varepsilon$.
%\end{proof}

\section{Parseval frames}\label{sec:Parseval_frames}

In this section, we want to prove that any exponentially localized  Parseval frame of $P_0$ can be transformed into an exponentially localized Parseval frame of $P_\varepsilon$ or $P_{k,\varepsilon}$ for $\varepsilon$ near zero. 

Assume that there is a set $\{\Psi_j\}_{j\in\{1,\dots ,m \}} \subset \mathcal{H}$ of vectors which are exponentially localized near the origin, i.e., there exist two constants $C,\alpha$ such that $\|\Psi_j(\gamma)\|<Ce^{-\alpha |\gamma|}$. 
Assume also that this set can be used to construct a Parseval frame of $P_0$, meaning that 
\begin{equation}\label{oct15}
P_0= \sum_{\gamma''\in q\Z\times \Z}\sum_{j=1}^m|\tau_{b_0,\gamma''}\Psi_{j}\rangle\langle \tau_{b_0,\gamma''} \Psi_{j}|.
\end{equation}
This implies that the hopping matrices are $$P_0(\gamma,\gamma') = \sum_{\gamma''\in q\Z\times \Z}\sum_{j=1}^m e^{ib_0\varphi (\gamma-\gamma',\gamma'')} \Psi_j(\gamma-\gamma'')\Psi_j(\gamma'-\gamma'')^*,$$
where $\Psi$ is a column vector, so $\Psi\Psi^*$ is actually a matrix.
Now, consider the operator 
$$\widetilde{P}_\vep\coloneq \sum_{\gamma''\in q\Z\times \Z}\sum_{j=1}^m|\tau_{b_0+\vep,\gamma''}\Psi_{j}\rangle\langle \tau_{b_0+\vep,\gamma''}\Psi_{j}|$$ 
with hopping matrices$$  \quad \widetilde{P}_\vep(\gamma,\gamma') =  \sum_{\gamma''\in q\Z\times \Z}\sum_{j=1}^m e^{i(b_0+\varepsilon)\varphi(\gamma-\gamma',\gamma'')} \Psi_j(\gamma-\gamma'')\Psi_j(\gamma'-\gamma'')^*.
$$
We have that
\[\begin{split}\left[\widetilde{P}_\varepsilon -\O_\varepsilon(P_0) \right]&(\gamma,\gamma') \\ &=\sum_{\gamma''\in q\Z\times \Z}\sum_{j=1}^m \left[e^{i\varepsilon\varphi(\gamma-\gamma',\gamma'')}- e^{i\varepsilon\varphi(\gamma,\gamma')}\right]e^{ib_0 \varphi (\gamma-\gamma',\gamma'')} \Psi_j(\gamma-\gamma'')\Psi_j(\gamma'-\gamma'')^*.\end{split}\]
Since for $x,y\in\R$ it holds that $|e^{ix}-e^{iy}|\le |x-y|$, and due to the identity $$\varphi(\gamma-\gamma',\gamma'')-\varphi(\gamma,\gamma')=\varphi(\gamma-\gamma',\gamma''-\gamma')$$ we can estimate
$$\left\|\left[\widetilde{P}_\varepsilon -\O_\varepsilon(P_0) \right](\gamma,\gamma')\right\| \le |\varepsilon| \, C^2\sum_{\gamma''\in q\Z\times \Z}\sum_{j=1}^m |\gamma-\gamma'|\, |\gamma'-\gamma''|\, e^{-\alpha|\gamma-\gamma''|}e^{-\alpha|\gamma''-\gamma'|},$$
which leads to 
$$\widetilde{P}_\varepsilon = \O_\varepsilon(P_0) + \mathcal{O}^{\exp}(\varepsilon) = P_{1,\varepsilon}+ \mathcal{O}^{\exp}(\varepsilon).$$ 
By construction, $\widetilde{P}_\varepsilon$ is $(b_0+\varepsilon,q)$-periodic and the previous equation also implies that $\widetilde{P}_\varepsilon ^2-\widetilde{P}_\varepsilon  = \mathcal{O}^{\exp}(\varepsilon) $, meaning that for $\varepsilon$ near zero, the trick implemented in Equation \eqref{eq:definition_P_k_epsilon} allows the definition of the following projection
$$Q_\vep=\frac{1}{2\pi i} \oint_{|z-1|=1/2} \left(z-\widetilde{P}_\vep\right)^{-1}dz.$$
As usual, it is $(b_0+\varepsilon,q)$-periodic and $Q_\varepsilon = \widetilde{P}_\varepsilon +  \mathcal{O}^{\exp}(\varepsilon) $. Define the local, self-adjoint operator
$$B_\vep\coloneq  \frac{1}{2\pi i} \oint_{|z-1|=1/2} { z^{-1/2}}\left(z-\widetilde{P}_\vep\right)^{-1}dz.$$
Then
$$Q_\vep =B_\vep \widetilde{P}_\vep B_\vep=\sum_{\gamma\in q\Z\times \Z}\sum_{j=1}^m|\tau_{b_0+\vep,\gamma}B_\vep \Psi_{j}\rangle\langle \tau_{b_0+\vep,\gamma}B_\vep\Psi_{j}|.$$
This shows that $\tau_{b_0+\vep,\gamma}B_{\vep} \Psi_j$ is an exponentially localized Parseval frame of $Q_\varepsilon$. Finally, since $P_{k,\vep}-Q_\vep=\mathcal{O}^{\exp}(\vep)$, for $\varepsilon$ near zero they are intertwined by a local, $(b_0+\varepsilon,q)$-periodic Kato-Nagy unitary $U_{k,\vep}$ such that $P_{k,\vep}=U_{k,\vep} Q_\vep U_{k,\vep}^*$, hence 
$$P_{k,\vep} =\sum_{\gamma\in q\Z\times \Z}\sum_{j=1}^m |\tau_{b_0+\vep,\gamma}U_{k,\vep} B_\vep \Psi_{j}\rangle\langle \tau_{b_0+\vep,\gamma}U_{k,\vep} B_\vep\Psi_{j}|.$$

The same can be done for $P_\varepsilon$ as $\|P_\varepsilon-\O_\varepsilon(P_0)\|\to0$, so there is a Kato-Nagy unitary $U_\varepsilon$ such that $P_{\vep}=U_{\vep} Q_\vep U_{\vep}^*$, therefore
$$P_{\vep} =\sum_{\gamma\in q\Z\times \Z}\sum_{j=1}^m |\tau_{b_0+\vep,\gamma}U_{\vep} B_\vep \Psi_{j}\rangle\langle \tau_{b_0+\vep,\gamma}U_{\vep} B_\vep\Psi_{j}|.$$
We note though that here the exponential decay insured by $U_\vep$ is inherited from the one of $P_\vep$, which in the locally gapped case might become weaker and weaker as $\vep\to 0$.
%In general, $\tau_{b_0+\varepsilon,\gamma}U_{\vep} B_\vep \Psi_{j}$ form a Parseval frame of $P_{\varepsilon}$, while $ \tau_{b_0+\varepsilon,\gamma}U_{k,\vep} B_\vep \Psi_{j}$ form a Parseval frame of $P_{k,\varepsilon}$.

\appendix

\section{Fiber decomposition in the presence of a rational magnetic flux}\label{sec:fiber_decomposition}

Assume that the magnetic flux is rational, namely $b_0=\frac{2\pi p}{q}$ where $p$ and $q$ are coprime integers.
Define 
\begin{equation*}
    T_{\alpha} = (-1)^{p\alpha_1\alpha_2} \tau_{b_0, (q\alpha_1,\alpha_2)}, \qquad \forall\alpha\in\Z^2.
\end{equation*}
The composition rule $$\tau_{b_0,\gamma}\circ\tau_{b_0,\gamma'}= e^{-ib_0\varphi(\gamma,\gamma')} \tau_{b_0,\gamma+\gamma'} = (-1)^{\left[\frac{p}{q}(\gamma_1\gamma_2'-\gamma_2\gamma_1')\right]}  \tau_{b_0, \gamma+\gamma'}$$
implies that $T_\alpha\circ T_{\alpha'} = T_{\alpha+\alpha'}$, so the family $T_\alpha$ constitutes a unitary representation of the family $ \tau_{b_0, (q\alpha_1,\alpha_2)}$.
In this process, we lose information on commutation with translations $ \tau_{b_0, (j,0)}$ for $j\in\{1,\dots ,q-1\}$, essentially, we are considering a larger unit cell with $n q$ elements.
Denoting the $a^{th}$ component of a vector $v$ as $v|_a$, we can use the Gelfand transform 
\begin{equation*}
    \mathcal{U}\colon\ell^2(\Z^2,\C^n) \rightarrow L^2(\mathbb{T}^2,dk,\C^{n q}) \quad \mbox{s.t.} \quad \forall \ 0\le j\le q-1, \  1\le a\le n,\  k \in\mathbb{T}^2 
\end{equation*}

\begin{equation*}
    (\mathcal{U}f)(k)|_{nj+a} = \frac{1}{2\pi} \sum_{\alpha\in\Z^2} T_\alpha(f)(j,0)|_a \cdot e^{-ik\cdot\alpha}
\end{equation*}
It is a unitary operator because $\mathcal{U}(\delta_{\gamma,a})$ is a vector with all components equal to zero except for the component $j^{th}$ which is equal to the function $\frac{1}{2\pi} e^{i\left[k_1\frac{\gamma_1-j}{q}+k_2\gamma_2\right ] }$ where $j$ is the unique number in $\{0,\dots, q-1\}$ such that $\gamma_1-j\in q\Z$.
By definition, $\mathcal{U}T_\alpha \mathcal{U}^{-1}$ is the multiplication operator $e^{ik\cdot \alpha}$, so an operator $A$ commutes with all $T_\alpha$ if and only if $\mathcal{U} A\mathcal{U}^{-1}$ admits a fiber decomposition, which is a family $a(k)$ of $nq\times nq$ matrices such that $(\mathcal{U} A\mathcal{U}^{-1}f)(k) = a(k)f(k)$.
The fibers can be expressed in terms of the hopping matrices of $A$ as follows
\begin{equation*}\label{eq:fiber_of_periodic_operator}a(k)|_{jn+a,j'n+a'} = \sum_{\alpha\in\Z^2} e^{ik\cdot \alpha} e^{-i\pi\frac{p}{q}\alpha_2(j'+q\alpha_1) }\left.A\big((j,0),(j'+q\alpha_1,\alpha_2)\big)\right|_{a,a'} \ .
\end{equation*}
The standard correspondence between real analytic functions and exponentially decreasing Fourier coefficients leads to the fact that the fibers are real analytic if and only if 
$$\left\|\left.A\big((j,0),(j'+q\alpha_1,\alpha_2)\big)\right|_{a,a'} \right\|\le C e^{-w|\alpha|}$$
for two constants $C,w>0$ that depend on $a,a',j,j'$. Since there is a finite number of parameters, the equation holds if and only if there are two fixed constants $C',w'>0$ such that 
$$\left\|\left.A\big((j,0),(j'+q\alpha_1,\alpha_2)\big)\right|_{a,a'} \right\|\le C' e^{-w'| (j'+q\alpha_1-j, \alpha_2)|}.$$
This justifies the validity of the following Lemma:
\begin{lemma}\label{lem:local+periodic=analytic_fibers}
    An operator $A$ is local and $(b_0,q)$-periodic if and only if $\mathcal{U} A \mathcal{U}^{-1}$ is decomposable in real analytic fibers.
\end{lemma}

We can apply the previous lemma to $H_{b_0}$ and have that its spectral properties are intertwined with the eigenvalues of its fibers $h(k)$ as $\sigma\left(H_{b_0}\right) = \bigcup_{k\in\mathbb{T}^2}\sigma (h(k))$.
If we label the eigenvalues of $h(k)$ in increasing order $\lambda_1(k)\le\cdots\le\lambda_{nq}(k)$, there are three scenarios that occur:
\begin{enumerate}
    \item Crossing eigenvalues: For some $\tilde{k}$ two eigenvalues are equal $\lambda_a(\tilde{k}) = \lambda_{a+1}(\tilde{k})$. 
    
    \item Separate eigenvalues: Two eigenvalues never touch $\lambda_a(k) < \lambda_{a+1}(k)$ for all $k\in \mathbb{T}^2$, but the overall spans are not separated.
    
    \item Strictly gapped eigenvalues: There is a value $c\in\R$ such that $\lambda_a(k) < c<\lambda_{a+1}(k)$ for all $k\in \mathbb{T}^2$.
\end{enumerate}

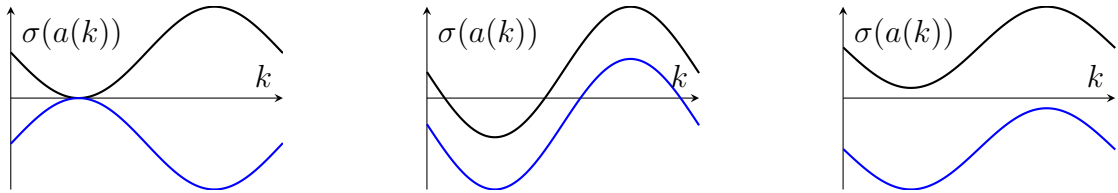
\begin{figure}[h]
\centering

\begin{minipage}{0.32\textwidth}
\centering
\begin{tikzpicture}
\begin{axis}[
    width=\linewidth,
    height=4cm,
    axis lines=middle,
    xlabel={$k$},
    ylabel={$\sigma(a(k))$},
    xtick=\empty,
    ytick=\empty,
    domain=0:6.3,
    samples=100,
]
\addplot[black, thick] {0.5-0.5*sin(deg(x))};
\addplot[blue, thick] {-0.5+0.5*sin(deg(x))};
\end{axis}
\end{tikzpicture}
\end{minipage}
\hfill
\begin{minipage}{0.32\textwidth}
\centering
\begin{tikzpicture}
\begin{axis}[
    width=\linewidth,
    height=4cm,
    axis lines=middle,
    xlabel={$k$},
    ylabel={$\sigma(a(k))$},
    xtick=\empty,
    ytick=\empty,
    domain=0:6.3,
    samples=100,
]
\addplot[blue, thick] {-0.2-0.5*sin(deg(x))};
\addplot[black, thick] {+0.2-0.5*sin(deg(x))};
\end{axis}
\end{tikzpicture}
\end{minipage}
\hfill
\begin{minipage}{0.32\textwidth}
\centering
\begin{tikzpicture}
\begin{axis}[
    width=\linewidth,
    height=4cm,
    axis lines=middle,
    xlabel={$k$},
    ylabel={$\sigma(a(k))$},
    xtick=\empty,
    ytick=\empty,
    domain=0:6.3,
    samples=100,
]
\addplot[black, thick] {0.5 - 0.4*sin(deg(x))};
\addplot[blue, thick] {-0.5 - 0.4*sin(deg(x))};
\end{axis}
\end{tikzpicture}
\end{minipage}

\caption{Representation of the three scenarios in one-dimension}
\end{figure}

If for $a,b$ either $a=1$ or $(\lambda_{a-1},\lambda_a)$ are separated and either $a+b=nq$ or $(\lambda_{a+b},\lambda_{a+b+1})$ are separated, then we call $\lambda_a\le \cdots\le \lambda_{a+b}$ a \emph{local spectral island} of $H_{b_0}$. Similarly, if we have strictly gapped pairs instead of separate pairs, then we call $\lambda_a\le \cdots\le \lambda_{a+b}$ a \emph{strict spectral island}.

\begin{lemma}\label{lem:P_0_local}
    Suppose that a real analytic matrix-valued function $h\colon\mathbb{T}^2\to M^{nq,nq}(\C)$ has a set of crossing eigenvalues $\lambda_a(k)\le \cdots\le \lambda_{a+b}(k)$ such that the pairs $(\lambda_{a-1},\lambda_a)$ and $(\lambda_{a+b},\lambda_{a+b+1})$ are separated.
    If $p\colon\mathbb{T}^2\to M^{nq,nq}(\C)$ is a matrix-valued function such that $p(k)$ projects into the combined eigenspaces of the selected crossing eigenvalues, then it is a real analytic function.
\end{lemma}

\begin{proof}
    We need a finite partition of unity on the torus $k \in \mathbb{T}^2$ $$\sum_{j=1}^L g_j(k) \equiv 1$$
    with an associated set of closed simple curves $\mathcal{C}_j \subset \C$ with the following property: if $k$ belongs to the support of $g_j$, then the curve $\mathcal{C}_j$ encloses the eigenvalues $\lambda_a(k)\le \cdots\le \lambda_{a+b} (k)$, leaving the other eigenvalues outside.
    Since it is always possible to construct such a partition, we have the following equality:
    \begin{equation}\label{eq:fibers_of_P_0}
        p(k)=\sum_{j=1}^L g_j(k) \frac{1}{2\pi i}\oint_{\mathcal{C}_j} (z-h(k))^{-1} dz .
    \end{equation}
    For every $k_0\in\mathbb{T}^2$ it is possible to select a small neighborhood $k_0 \in N\subset \mathbb{T}^2$ and two constants $c,d\in\R$ such that $\lambda_{a-1}(k)<c<\lambda_a(k)$ and $\lambda_{a+b}(k)<d<\lambda_{a+b+1}(k)$ for all $k\in N$.
    Then we can define a closed, simple curve $\mathcal{N}\subset \C$ enclosing the interval $[c,d]$ that crosses the real axis exactly at points $c,d$, thus enclosing all the eigenvalues $\lambda_a(k)\le \cdots\le \lambda_{a+b} (k)$ for $k\in N$.
    In points $k\in N$, it holds that $$\sum_{j=1}^L g_j(k) \frac{1}{2\pi i}\oint_{\mathcal{C}_j} (z-h(k))^{-1} dz = \sum_{j=1}^L g_j(k) \frac{1}{2\pi i}\oint_{\mathcal{N}} (z-h(k))^{-1} dz = \frac{1}{2\pi i}\oint_{\mathcal{N}} (z-h(k))^{-1} dz.$$
    So the regularity of $p(k)$ is the same as $h(k)$, which is real analytic.
\end{proof}

We can use this Lemma to define $P_0 (f) = \mathcal{U}^{-1}[p(k) (\mathcal{U}f)(k)]$, a projection associated to a spectral island of $H_{b_0}$. Since each fiber commutes with $h(k)$, we have $\left[P_0,H_{b_0}\right]=0$, but $P$ may not be a spectral projection of $H_{b_0}$.

In the strictly gapped case, one can select $c,d$ such that $\lambda_{a-1}(k)<c<\lambda_a(k)$ and $\lambda_{a+b}(k)<d<\lambda_{a+b+1}(k)$ for all $k\in \mathbb{T}^2$, so we can avoid the construction with the partition of unity and have that 
\begin{equation}\label{eq:strictly_gaped_zero_projection}
    p_{0}(k)=\frac{1}{2\pi i}\oint_{\mathcal{C}} (z-h(k))^{-1} dz.
\end{equation}
In this case $P_0=\Id_{[c,d]}\left(H_{b_0}\right)$ where $\Id$ denotes the characteristic function, hence $P_0$ is the spectral projection of $H_{b_0}$.

\section{Exponential localization}\label{sec:exp_localization}
\begin{lemma}\label{lem:locality_VS_sums_compositions}
    Sums and compositions of operators in $\mathcal{B}^{\exp}(\mathcal{H})$ are still in $\mathcal{B}^{\exp}(\mathcal{H})$. 
\end{lemma}

\begin{proof}
    Consider two operators $A,B \colon\mathcal{H} \to \mathcal{H}$ whose hopping matrices obey the constraints $$\|A(\gamma,\gamma')\|\le Ce^{-a |\gamma-\gamma'|},\qquad \|B(\gamma-\gamma')\| \le De^{-b |\gamma-\gamma'|}. $$
    At the cost of having a less strict constraint, we can always have $a <b $. Then their sum has hopping matrices such that $$ \|A(\gamma,\gamma')+B(\gamma,\gamma')\|\le (C+D)e^{-a |\gamma-\gamma'|}.$$
    For their product, we have
    $$\left\|(A\circ B)(\gamma,\gamma')\right\|=\left\|\sum_{\gamma''\in\Z^2}A(\gamma,\gamma'')B(\gamma'',\gamma')\right\|\le CD \sum_{\gamma''\in\Z^2} e^{-a  |\gamma-\gamma''|-b |\gamma''-\gamma'|}.$$
    Up to a translation $\gamma'' = \gamma''+\gamma'$, the last term is equal to $$CD \sum_{\gamma''\in\Z^2} e^{-a  |\gamma-\gamma'-\gamma''|-b |\gamma''|}\le CD e^{-a  |\gamma-\gamma'|} \sum_{\gamma''\in\Z^2} e^{-(b -a ) |\gamma''|}.$$
Since the last series converges to a real number, we have completed the proof.
\end{proof}

\begin{lemma}\label{lem:locality_VS_continuity}
    Local operators are also continuous ($\mathcal{B}^{\exp}(\mathcal{H })\subset \mathcal{B}(\mathcal{H})$), in particular, if $$\|A(\gamma,\gamma')\|\le Ce^{-a |\gamma-\gamma'|} \quad \mbox{and} \quad \sum_{\gamma\in\Z^2} e^{-a |\gamma|} = S_a ,\quad \mbox{then} \quad \|A\|\le CS_a .$$
\end{lemma}

\begin{proof}
    It is a consequence of the Schur test to local operators. 
    %Start with
    %$$\|A(\gamma,\gamma') f(\gamma')\| \le \|A(\gamma,\gamma')\|\|f(\gamma')\| \le \sqrt{\|A(\gamma,\gamma')\|} \sqrt{\|A(\gamma,\gamma')\| \|f(\gamma)\|^2}.$$
    %The Cauchy-Schwarz inequality implies that:
    %$$\|(Af)(\gamma)\|^2 = \left\|\sum_{\gamma'\in\Z^2} A(\gamma,\gamma')f(\gamma') \right\|^2\le \left(\sum_{\gamma'\in\Z^2} \|A(\gamma,\gamma')\|\right)\left(\sum_{\gamma'\in\Z^2} \|A(\gamma,\gamma')\| \|f(\gamma')\|^2\right)$$
    %The first factor can be estimated with $CS_a $, so:
    %\begin{align*}
     %   \|Af\|^2  &= \sum_{\gamma\in\Z^2} \|(Af)(\gamma)\|^2\le CS_a  \sum_{\gamma,\gamma'\in\Z^2} \|A(\gamma,\gamma')\| \|f(\gamma')\|^2 \\ 
        %&\le CS_a  \sum_{\gamma'\in\Z^2} \|f(\gamma')\|^2 \sum_{\gamma\in\Z^2}\|A(\gamma,\gamma')\|\le (CS_a )^2 \|f\|^2\tag*{\qedhere}
    %\end{align*}
\end{proof}

%\begin{example}
 %   Fix $n=1$ and consider the operators $A,A_l\in\mathcal{B}(\mathcal{H})$ with hopping matrices
  %  $$A(\gamma,\gamma')=\begin{cases}
   %     \frac{1}{|\gamma|+1} &\mbox{if }\gamma'=-\gamma, \\0&\mbox{otherwise,}
    %\end{cases} \quad A_l(\gamma,\gamma')=\begin{cases}
     %   \frac{1}{|\gamma|+1} &\mbox{if }\gamma'=-\gamma \mbox{ and }\|\gamma\|\le l,\\0&\mbox{otherwise.}
    %\end{cases}$$
   % For every $l\in\N$, $A_l$ has a finite number of non-trivial hopping matrices, so $A_l\in \mathcal{B}^{\exp}(\mathcal{H })$.
    %Moreover, $\lim_{l\to\infty} \|A-A_l\|=0$, however, $A\notin \mathcal{B}^{\exp}(\mathcal{H })$. 
%\end{example}

\begin{lemma}\label{lem:locality_of_resolvent}
    Let $A\in \mathcal{B}^{\exp}(\mathcal{H})$ be a local operator such that $\|A(\gamma,\gamma')\|\le C e^{-a  |\gamma-\gamma'|}$ and fix $z\not\in \sigma(A)$. Then $(A-z)^{-1}\in \mathcal{B}^{\exp}(\mathcal{H})$.
    \end{lemma}
    \begin{proof}
        Fix $0<d<a $ and $\gamma_0\in \Z^2$. Let $A_{d,\gamma_0}$ be defined via its hopping matrices 
        $$A_{d,\gamma_0}(\gamma,\gamma')=e^{d[|\gamma-\gamma_0|-|\gamma'-\gamma_0|]} A(\gamma,\gamma'). $$
        We have 
        $$\Vert A_{d,\gamma_0}(\gamma,\gamma')\Vert\leq C e^{-(a -d)|\gamma-\gamma'|}$$
        and Lemma \ref{lem:locality_VS_continuity} implies that $\|A_{d,\gamma_0}\|$ is bounded uniformly in $\gamma_0$. Using $|e^x-1|\leq e^{|x|}-1$ for all $x\in \R$ we have 
        \[\begin{split}\|A_{d,\gamma_0}(\gamma,\gamma')-A(\gamma,\gamma')\| &\leq \left|e^{d[|\gamma-\gamma_0|-|\gamma'-\gamma_0|]}-1\right| C\, e^{-a |\gamma-\gamma'|} \\
        & \le  \left|e^{d|\gamma-\gamma'|}-1\right| C\, e^{-a |\gamma-\gamma'| }= \left| 1-e^{-d|\gamma-\gamma'|}\right|     C e^{-(a-d) |\gamma-\gamma'| } \\&\le d|\gamma-\gamma'| C\, e^{-(a-d) |\gamma-\gamma'|}. \end{split}\]
        Up to rescaling $a '$ and $C'$, we can have $$|\gamma-\gamma'| C\, e^{-a |\gamma-\gamma'|}\le C' e^{-a '|\gamma-\gamma'|}, $$ so $A_{d,\gamma_0}-A=\mathcal{O}^{\exp}(d)$ uniformly in $\gamma_0$. 

        The norm of $(A-z)^{-1}\mathcal{O}^{\exp}(d)$ is of order 
        $d \Vert (A-z)^{-1}\Vert$
        hence it goes to zero with $d$,  uniformly in $\gamma_0$, allowing us to define the inverse of $\Id +(A-z)^{-1}\mathcal{O}^{\exp}(d)$ using the Neumann series for $d$ small enough. 
        
        Then, if $d=\delta(z)$ is small enough such that $\Vert  \mathcal{O}^{\exp}(d)\Vert \, \Vert (A-z)^{-1}\Vert<1/2$, the inverse of $A_{\delta(z),\gamma_0}-z$ can be calculated as follows 
        $$(A_{\delta(z),\gamma_0}-z)^{-1}=\big (\Id +(A-z)^{-1}\mathcal{O}^{\exp}(d)\big )^{-1} (A-z)^{-1},\quad \|(A_{\delta(z),\gamma_0}-z)^{-1}\| \le 2\Vert (A-z)^{-1}\Vert.$$

        Consider the operator $e^{-\delta(z) |\cdot -\gamma_0|} $ with hopping matrices $$ e^{-\delta(z) |\cdot -\gamma_0|}(\gamma,\gamma')=\begin{cases}
            \Id \cdot e^{-\delta(z) |\gamma -\gamma_0|}& \mbox{for }\gamma=\gamma', \\ 0&\mbox{otherwise}.
        \end{cases} $$ 
        Then, on compactly supported functions it holds that 
        $$(A-z)^{-1}=e^{-\delta(z)|\cdot -\gamma_0|}(A_{\delta(z),\gamma_0}-z)^{-1}e^{\delta(z)|\cdot -\gamma_0|}$$
        which leads to 
        $$(A-z)^{-1}e^{-\delta(z)|\cdot -\gamma_0|}=e^{-\delta(z)|\cdot -\gamma_0|}(A_{\delta(z),\gamma_0}-z)^{-1}$$
        and finally to
        $$e^{\delta(z)|\cdot -\gamma_0|}(A-z)^{-1}e^{-\delta(z)|\cdot -\gamma_0|}=(A_{\delta(z),\gamma_0}-z)^{-1}.$$
        It shows that 
        $$\left\|e^{\delta(z)|\cdot -\gamma_0|}(A-z)^{-1}e^{-\delta(z)|\cdot -\gamma_0|}\right\| \le 2\Vert (A-z)^{-1}\Vert$$
        is uniformly bounded in $\gamma_0$. This implies that its hopping matrices obey the following estimate 
        $$e^{\delta(z)|\gamma -\gamma_0|} \left\|(A-z)^{-1}(\gamma,\gamma_0)\right\| \le C \Vert (A-z)^{-1}\Vert $$
        for some constant $C>0$ and this concludes the proof. 
    \end{proof}

\begin{corollary}\label{cor:locality_VS_functional_calculus}
Let $A\in \mathcal{B}^{\exp}(\mathcal{H})$ be self-adjoint. Let $f$ be a complex-valued function which is analytic on a domain including the spectrum of $A$ and define $f(A)$ through analytic calculus. Then $f(A)\in \mathcal{B}^{\exp} (\mathcal{H})$. 
\end{corollary}
\begin{proof}
    Let $\mathcal{C}\subset \C$ be a closed simple contour that lies in the analyticity domain of $f$ and containing the entire spectrum of $A$.
    %For such a function, the Cauchy formula $f(w)=\frac{1}{2 \pi i} \oint_\mathcal{C} dz \frac{f(z)}{z-w} dz$ holds for all $w$ within $\mathcal{C}$. So we have
    
    The exponential decay rate $\delta(z)>0$ of the matrix elements of $(A-z)^{-1} $ from Lemma \ref{lem:locality_of_resolvent} can be chosen uniformly in the distance between $z\in \mathcal{C}$ and the spectrum of $A$ because $A$ is self-adjoint. This implies that
    $$f(A)=\frac{1}{2 \pi i}\oint_\mathcal{C} f(z)\, (z-A)^{-1}\, dz$$
    has also exponentially localized matrix elements. 
    %If we switch to the hopping matrices, we have that:
    %$$\|f(A)(\gamma,\gamma')\| \le\frac{1}{2 \pi}\oint_\mathcal{C} \|f(z)\|\, \left\|(z-A)^{-1}\right\|\, dz\le  \frac{1}{2\pi} \oint_\mathcal{C} dz \|f(z)\|\left\|e^{-\delta(z)|\gamma -\gamma_0|}\right\|  D_{\delta(z)} N_{\delta(z),z}$$
    %We can choose $0<\delta'<\delta(z)$, $D>D_{\delta(z)}$ and $ N> N_{\delta(z)}$ for all $z \in$ to obtain the final estimate
    %\begin{align*}
     %   \|f(A)(\gamma,\gamma')\| \le \frac{\|\mathcal{C}\| \max_\mathcal{C} |f|}{2\pi} DN  e^{-\delta'|\gamma -\gamma'|}.\tag*{\qedhere}
    %\end{align*}  
\end{proof}

\begin{small}
    \printbibliography
    %\bibliographystyle{plainnat}
    %\bibliography{biblio}
\end{small}

\end{document}